\documentclass[%
reprint,
superscriptaddress,
nofootinbib,
bibnotes,
amsmath,
amssymb,
longbibliography,
aps,
pra,
́́]{revtex4-2}

\usepackage{graphicx}
\usepackage{dcolumn}
\usepackage{braket}
\usepackage{textcomp, gensymb}
\usepackage{bm}
\usepackage{amsmath}
\usepackage{float}
\usepackage{mathtools}
\usepackage{color}
\usepackage[dvipsnames]{xcolor}
\usepackage[utf8]{inputenc}
\usepackage[T1]{fontenc}
\usepackage{csquotes}
\usepackage{thmtools}
\usepackage{qcircuit}
\usepackage{tikz-cd}
\newtheorem{lemma}{Lemma}

\makeatletter
\DeclareRobustCommand{\qed}{%
  \ifmmode 
  \else \leavevmode\unskip\penalty9999 \hbox{}\nobreak\hfill
  \fi
  \quad\hbox{\qedsymbol}}
\newcommand{\openbox}{\leavevmode
  \hbox to.77778em{%
  \hfil\vrule
  \vbox to.675em{\hrule width.6em\vfil\hrule}%
  \vrule\hfil}}
\newcommand{\qedsymbol}{\openbox}
\newenvironment{proof}[1][\proofname]{\par
  \normalfont
  \topsep6\p@\@plus6\p@ \trivlist
  \item[\hskip\labelsep\itshape
    #1.]\ignorespaces
}{%
  \qed\endtrivlist
}
\newcommand{\proofname}{Proof}
\makeatother

\renewcommand{\thesubsection}{\thesection.\Roman{subsection}}

\usepackage{url}

\usepackage{breakurl}
\usepackage[breaklinks]{hyperref}
\hypersetup{pdfborder=0 0 0}

\usepackage[normalem]{ulem}

\usepackage{physics}
\usepackage[version=4]{mhchem}
\usepackage{dsfont}

\usepackage{natbib}
\usepackage{empheq}
\usepackage{cases}
\usepackage[all]{hypcap}
\usepackage{siunitx}
\usepackage[capitalize]{cleveref}
\usepackage{placeins}
\usepackage{listings}

\usepackage{tabularx}
\usepackage{nicematrix}
\usepackage{stmaryrd}
\usepackage{footmisc}

\crefname{section}{Sec.}{Secs.}

\lstdefinestyle{inline}{basicstyle=\tt}

\DeclareSIUnit{\Hartree}{Ha}

\definecolor{myorange}{HTML}{FE6100}
\definecolor{myblue}{HTML}{64A5FF}
\definecolor{mypink}{HTML}{FF1972}
\definecolor{mygreen}{HTML}{62D09D}
\definecolor{myviolet}{HTML}{4F428E}
\definecolor{myyellow}{HTML}{FFBF00}

\renewcommand{\thesubsection}{\Roman{section} \Alph{subsection}}

\makeatletter
\def\p@subsection{}
\makeatother
\makeatletter
\def\p@subsubsection{}
\makeatother

\allowdisplaybreaks

\newcommand{\pvec}[1]{\vec{#1}\mkern2mu\vphantom{#1}}
\def\tauRange{\mathcal T}
\def\tRange{\left[T\right]}
\def\qubits{\left\llbracket n\right\rrbracket}

\begin{document}

\title{
From virtual \texorpdfstring{$Z$}{Z} gates to virtual \texorpdfstring{$Z$}{Z} pulses
}

\author{Christopher K. Long}
\email{ckl45@cam.ac.uk}
\affiliation{Cavendish Laboratory, Department of Physics, University of Cambridge, JJ Thomson Avenue, Cambridge, CB3 0US, United Kingdom}
\affiliation{Hitachi Cambridge Laboratory, J. J. Thomson Avenue, Cambridge, CB3 0US, United Kingdom}
\author{Crispin H. W. Barnes}
\affiliation{Cavendish Laboratory, Department of Physics, University of Cambridge, JJ Thomson Avenue, Cambridge, CB3 0US, United Kingdom}

\date{\today}

\begin{abstract}
Virtual $Z$ gates have become integral for implementing fast, high-fidelity single-qubit operations. However, virtual $Z$ gates require that the system's two-qubit gates are microwave-activated or normalise the single-qubit $Z$ rotations---the group generated by $X$, $\operatorname{SWAP}$, and arbitrary phase gates. Herein, we extend the theory of virtual $Z$ gates to the pulse-level, which underlies both gate design and the recent advancements of pulse-level quantum algorithms. These algorithms attempt to utilise the full potential of present-day noisy intermediate-scale quantum (NISQ) devices by removing overheads associated with the compilation and transpilation of gates. To extend the theory of virtual $Z$ gates, we derive a platform-agnostic theoretical framework for virtual $Z$ pulses by employing time dilations of the pulse sequences that control the quantum processor. Additionally, we provide worked examples of the implementation of virtual $Z$ pulses on both semiconductor spin qubit and superconducting quantum processor architectures. Moreover, we present a general overview of the hardware support for virtual $Z$ pulses. We find virtual $Z$ pulses (and thus, virtual $Z$ gates) can be used on hardware that, with previous methods, did not support the virtual $Z$ gate. Finally, we present two additional applications of virtual $Z$ pulses to pulse-level algorithms. First, broadening the class of Hamiltonians that can be natively simulated in an analogue manner. Second, increasing the expressibility of pulse-based variational quantum algorithms.
\end{abstract}

\maketitle

\section{Introduction}

Digital quantum computing frameworks frequently employ virtual $Z$ gates \cite{PhysRevA.96.022330,PhysRevB.83.121403,PhysRevA.77.012307,RevModPhys.76.1037} to obtain fast, high-fidelity single- \cite{Li2023,STASIUK2024107688} and two-qubit \cite{Danin_2025} operations. Virtual $Z$ gates can be implemented instantaneously by updating the rotating frame, have unit fidelity, and allow universal quantum computation \cite{PhysRevA.96.022330} when supplemented with only an $X$ $\frac{\pi}{2}$-rotation and a single entangling gate \cite{PhysRevLett.75.346,PhysRevA.52.3457}.

Virtual $Z$ gates and the general method of frame tracking \cite{PRXQuantum.5.020338,PhysRevResearch.6.013235} have been demonstrated on a variety of hardware platforms: such as in NMR \cite{Knill2000368,STASIUK2024107688}, superconducting qubits \cite{PhysRevA.96.022330,Li2023,PhysRevResearch.3.L042007,PhysRevLett.126.210504,PhysRevLett.132.060602,PhysRevX.13.031035,Moskalenko2022,PhysRevX.11.021026,PhysRevResearch.4.023040,PhysRevLett.127.200502,PhysRevApplied.14.044039,PhysRevX.14.041050,PhysRevLett.127.080505,PRXQuantum.4.010314,Ye_2021,PhysRevLett.125.240503,Xu_2021,Goss2022,PhysRevLett.129.060501,PhysRevResearch.2.033447,PhysRevX.11.041032,PRXQuantum.5.020326,PRXQuantum.5.020338,PRXQuantum.6.010349,PhysRevApplied.20.024011,Kim2022,PhysRevApplied.14.014072}, semiconductor spin qubits \cite{doi:10.1126/sciadv.abn5130,Wu2024,doi:10.1021/acs.nanolett.4c05540,Watson2018,wu2025simultaneoushighfidelitysinglequbitgates}, neutral atoms \cite{doi:10.1126/science.abo6587}, and trapped ions \cite{PhysRevA.77.012307,Postler2022}. Often, high-fidelity multi-qubit gates on these platforms employ the virtual $Z$ gate for phase corrections. Examples include CPHASE \cite{PhysRevLett.132.060602,PhysRevX.13.031035,Moskalenko2022,PhysRevX.11.021026,PhysRevResearch.4.023040,PhysRevLett.127.200502,PhysRevApplied.14.044039,PhysRevX.14.041050,PhysRevLett.127.080505,PRXQuantum.4.010314,Ye_2021,PhysRevLett.125.240503,Xu_2021,Goss2022,PhysRevLett.129.060501,PhysRevResearch.2.033447,PhysRevX.11.041032,doi:10.1126/sciadv.abn5130,Wu2024,Watson2018}, CCPHASE \cite{PhysRevApplied.14.014072,PhysRevA.98.052318}, CNOT \cite{PRXQuantum.6.010349,PhysRevLett.129.060501,PhysRevApplied.20.024011,Wu2024}, Toffoli \cite{PhysRevApplied.14.014072,PhysRevA.98.052318}, $i$Toffoli \cite{Kim2022}, SWAP \cite{PRXQuantum.5.020338,doi:10.1021/acs.nanolett.4c05540}, power of $i$SWAP \cite{PhysRevResearch.2.033447,PhysRevX.11.041032,PRXQuantum.5.020326,PRXQuantum.5.020338}, and power of $b$SWAP \cite{PRXQuantum.5.020326,PRXQuantum.5.020338}. Further, virtual $Z$ gates can be utilised in dynamical decoupling \cite{PRXQuantum.6.020348,PhysRevApplied.22.054074,PhysRevApplied.18.024068}, to mitigate cross-talk \cite{PhysRevApplied.18.024068,10.1063/5.0200014,10.1063/5.0115393,wu2025simultaneoushighfidelitysinglequbitgates}, and to enable gate set tomography protocols \cite{PhysRevLett.133.120802}.

In practice, digital quantum computers are underpinned by controllable quantum systems upon which digital quantum gates are engineered by careful pulse design \cite{Glaser2015,Mahesh2023,Müller_2022,Barnes_2022}. Additionally, there has been a recent and promising push to leverage the power of today's noisy intermediate-scale quantum (NISQ) \cite{Preskill2018quantumcomputingin} devices by designing algorithms at the pulse level \cite{doi:10.1126/science.abo6587,PRXQuantum.2.010101,lloyd2020quantumpolardecompositionalgorithm,lloyd2021hamiltoniansingularvaluetransformation,PhysRevApplied.23.024036,Meitei2021,PhysRevResearch.5.033159,9996174,10.3389/frqst.2023.1273581,PhysRevD.111.034506,PhysRevX.7.021027,PhysRevApplied.19.064071,PhysRevLett.118.150503,Lu2017,Long2025}. These pulse-level algorithms can be executed orders of magnitude faster than their gate-based counterparts \cite{PhysRevApplied.19.064071,Long2025}, thus minimising decoherence---a crippling issue for many gate-based NISQ algorithms \cite{Dalton2024,PhysRevA.109.042413,Wang2021,PRXQuantum.4.010309,StilckFrança2021}.

Additionally, the unit fidelity of virtual $Z$ gates makes them a useful tool for enabling gate-based NISQ algorithms. Unfortunately, as virtual $Z$ gates are gates, they are not compatible with pulse-based algorithms. Nonetheless, some hybrid pulse--gate based algorithms are already using virtual $Z$ gates \cite{PhysRevResearch.5.033159,doi:10.1126/science.abo6587}.

In this article, we extend the theory of virtual $Z$ gates to the pulse level: virtual $Z$ pulses. Our extended theory relies extensively on the application of time dilations. This effect is achieved by distorting the pulse sequence to effectively implement $Z$ pulses. Previous works have used time dilations to design robust gates \cite{Rimbach-Russ_2023,polat2025pulseshapingultrafastadiabatic}. Interestingly, more hardware platforms support virtual $Z$ pulses than gates. Virtual $Z$ gates require that two-qubit gates are either microwave-activated or normalise single-qubit $Z$ rotations.\footnote{The normalisers group $\mathcal N$ of a group $L\subseteq\operatorname{U}(2^n)$ is $\mathcal N\coloneqq\left\{u\in\operatorname{U}(2^n):uL=Lu\right\}$.} Instead, a weaker constraint on the Hamiltonian algebra arises for virtual $Z$ pulses. This allows us to derive a wider array of virtual $Z$--compatible two-qubit gates. We give explicit examples of the applications of virtual $Z$ gates in semiconductor spin qubits \cite{RevModPhys.95.025003} and tunable-coupler, flux-tunable, and cross-resonance superconducting qubits \cite{10.1063/1.5089550}. Further, we consider a couple of potential applications of virtual $Z$ pulses to pulse-level algorithms: variational quantum algorithms and Hamiltonian simulation. Our work extends the native controls available to both gate- and pulse-level quantum-algorithm design.

The remainder of this article will be structured as follows: \Cref{sec: results} outlines the mathematical framework abstracted away from a specific hardware platform. \Cref{sec: examples} outlines four specific hardware implementations: semiconductor spin qubit architectures (\cref{sec: semiconductor}) and superconducting qubits with tunable coupling (\cref{sec: tunable couplers}), tunable flux (\cref{sec: tunable flux}), and cross-resonance coupling (\cref{sec: cr}). In \cref{sec: gate design} we give a more general overview of the wider hardware support of virtual $Z$ pulses. Next, we discuss three applications: virtual $Z$ gates (\cref{sec: gate design}), variational quantum algorithms (\cref{sec: VQAs}), and Hamiltonian simulation (\cref{sec: Hamiltonian simulation}). Finally, we conclude in \cref{sec: conclusion}.

\section{Theory}\label{sec: results}

In this section, we outline the theory of virtual $Z$ pulses (\cref{sec: theory}) before applying the theory to 2-local Hamiltonians (\cref{sec: local Hamiltonian}).

\subsection{Theory of virtual \texorpdfstring{$Z$}{Z} pulses}
\label{sec: theory}
\subsubsection{The virtual frame}
Suppose we have a time-dependent system Hamiltonian $H\left(t\right)$, but desire to implement
\begin{equation}\label{eq: effective Hamiltonian}
    H_{\textrm{eff}}\left(t\right)\coloneqq H\left(t\right)+\frac{1}{2}\sum_{i=1}^{n}v_i\left(t\right)Z_i
\end{equation}
over the time interval $t\in\tRange\coloneqq\left[0, T\right]$, where $Z_i$ is the Pauli-$Z$ operator on the $i$th qubit. We will achieve this by virtually introducing the drive terms $\frac{1}{2}\sum_{i=1}^nv_i\left(t\right)Z_i$. The unitary that implements the desired time evolution is the solution to the Schr\"odinger equation:
\begin{equation}
  \dv{t}U\left(t\right)=-\frac{i}{\hbar}\left[H\left(t\right)+\frac{1}{2}\sum_{i=1}^{n}v_i\left(t\right)Z_i\right]U\left(t\right),\quad\forall t\in\tRange.
\end{equation}
Making the substitution
\begin{equation}
  U\left(t\right)=R\left(t\right)\tilde U\left(t\right),\quad\forall t\in\tRange,
\end{equation}
with
\begin{align}
    &R\left(t\right)\coloneqq \exp\left(-\frac{i}{2\hbar}\sum_{i=1}^{n}V_i\left(t\right)Z_i\right)\\\textrm{and }&V_i\left(t\right)\coloneqq V_{i0}+\int_0^t\dd{x}v_i\left(x\right)\quad\forall i\in\qubits,
\end{align}
where $\qubits$ is the set of positive integers 1 through $n$ inclusive, we move into \textit{the virtual frame}. In which we find
\begin{equation}\label{eq: simple differential equation}
  \dv{t}\tilde U\left(t\right)=-\frac{i}{\hbar}R^\dagger\left(t\right)H\left(t\right)R\left(t\right)\tilde U\left(t\right),\quad\forall t\in\tRange.
\end{equation}
The effective Hamiltonian in the virtual frame is
\begin{equation}\label{eq: simple Hamiltonian}
    \tilde H(t)\coloneqq R^\dagger\left(t\right)H\left(t\right)R\left(t\right),\quad\forall t\in\tRange.
\end{equation}
In other words, the virtual frame is the unique rotating frame, up to a constant rotation, within which the effective Hamiltonian is instantaneously unitarily equivalent to our system Hamiltonian $H\left(t\right)$ at all times. The constant rotation,
\begin{equation}
    \exp\left(-\frac{i}{2\hbar}\sum_{i=1}^{n}V_{i0}Z_i\right),
\end{equation}
allows prior virtual $Z$ gates to be pulled through the evolution.

Thus, providing that we physically implement $\tilde H(t)$ in the laboratory frame, then we can implement the evolution from a time $t=0$ to $t=T$ under the Hamiltonian $H_{\textrm{eff}}\left(t\right)$. This can be achieved as follows: Evolve the system with $\tilde H \left(t\right)$ from a time $t=0$ to $t=T$ to induce the unitary evolution $\tilde U \left(T\right)$. Then, Apply the $Z$ rotations $R(T)$. In many cases, these $Z$ rotations can be absorbed into the measurement operators or can be implemented as virtual $Z$ gates \cite{PhysRevA.96.022330,PhysRevB.83.121403,PhysRevA.77.012307,RevModPhys.76.1037}.

\begin{figure*}
    \centering
    \centerline{
    \Qcircuit @C=1em @R=2em {
        \bm{(a)}&                                                                                               &                                                                                               &     &&        && \bm{(b)}\\
                & t\in\left[0,T_1\right] \barrier[-6.05em]{3}                                                   & t\in\left[T_1,T_2\right]                                                                      &     &&        &&          & \tau\in\left[\tau_0, \tau_1\right] \barrier[-2.75em]{3} & t\in\left[0, T_1\right] \barrier[-2.5em]{3} & \tau\in\left[\tau_1, \tau_2\right] \barrier[-2.75em]{3} & t\in\left[T_1, T_2\right]   & \\
      \lstick{1}& \multigate{1}{\displaystyle H_{12}\left(t\right)+\frac{1}{2}\sum_{i=1}^2v_i\left(t\right)Z_i} & \gate{        \displaystyle H_{1 }\left(t\right)+\frac{1}{2}            v_1\left(t\right)Z_1} & \qw &&        &&          & \multigate{1}{H_{12}'\left(\tau\right)}                 & \gate{\frac{1}{2}v_1(t)Z_1}                 & \gate{        H_{1 }'\left(\tau\right)}                 & \gate{\frac{1}{2}v_1(t)Z_1} & \qw \\
      \lstick{2}& \ghost{       \displaystyle H_{12}\left(t\right)+\frac{1}{2}\sum_{i=1}^2v_i\left(t\right)Z_i} & \multigate{1}{\displaystyle H_{23}\left(t\right)+\frac{1}{2}\sum_{i=2}^3v_i\left(t\right)Z_i} & \qw && \equiv &&          & \ghost{       H_{12}'\left(\tau\right)}                 & \gate{\frac{1}{2}v_2(t)Z_2}                 & \multigate{1}{H_{23}'\left(\tau\right)}                 & \gate{\frac{1}{2}v_2(t)Z_2} & \qw \\
      \lstick{3}& \gate{        \displaystyle H_{3 }\left(t\right)+\frac{1}{2}            v_3\left(t\right)Z_3} & \ghost{       \displaystyle H_{23}\left(t\right)+\frac{1}{2}\sum_{i=2}^3v_i\left(t\right)Z_i} & \qw &&        &&          & \gate{        H_{3 }'\left(\tau\right)}                 & \gate{\frac{1}{2}v_3(t)Z_3}                 & \ghost{       H_{23}'\left(\tau\right)}                 & \gate{\frac{1}{2}v_3(t)Z_3} & \qw
    }}
    \vspace{2em}
    \centerline{
    \Qcircuit @C=1em @R=2em {
        \bm{(c)}\\
                & \tau\in\left[\tau_0, \tau_1\right] \barrier[-2.75em]{3} & \tau'\in\left[\tau_1', \tau_2'\right] \barrier[-2.75em]{3} & t\in\left[0, T_2\right]      & \\
                & \multigate{1}{H_{12}'\left(\tau\right)}                & \gate{        H_{1 }''\left(\tau'\right)}                  & \gate{\frac{1}{2}v_1(t)Z_1} & \qw \\
 \lstick{\equiv}& \ghost{       H_{12}'\left(\tau\right)}                & \multigate{1}{H_{23}''\left(\tau'\right)}                  & \gate{\frac{1}{2}v_2(t)Z_2} & \qw \\
                & \gate{        H_{3 }'\left(\tau\right)}                & \ghost{       H_{23}''\left(\tau'\right)}                  & \gate{\frac{1}{2}v_3(t)Z_3} & \qw
    }}
    \caption{\textbf{Circuit diagrams for an example virtual $Z$ pulse sequence.} Each box in the circuit diagram represents the unitary operator generated by the evolution of the system under the time-dependent Hamiltonian labelling the box. The time period over which the Hamiltonian is applied is annotated above each column. Boxes that act on multiple qubit lines imply the Hamiltonian contains an interaction term between the qubits. \textbf{(a)} A circuit diagram for the desired evolution. The native system Hamiltonian is $H_{12}(t)+H_3(t)$ for $t\in\left[0,T_1\right]$ and $H_1(t)+H_{23}(t)$ for $t\in\left[T_1,T_2\right]$. In addition, we wish to apply $\frac{1}{2}\sum_{i=1}^3v_i(t)Z_i$ for $t\in\left[0,T_2\right]$ virtually. In the time period $t\in\left[0,T_1\right]$ only qubits 1 and 2 interact while for $t\in\left[T_1,T_2\right]$ only qubits 2 and 3 interact. \textbf{(b)} By employing our virtual $Z$ method outlined in the Article we can factor out the virtual $Z$ terms by modifying the Hamiltonians to their primed counterparts and potentially employing a time-dilation: $t=f(\tau)$ such that $\tau_0,\tau_1,\tau_2=f^{-1}(0, T_1, T_2)$. The factored out $Z$ rotations could now either be implemented physically or as virtual $Z$ gates. \textbf{(c)} Alternatively, we can apply our virtual $Z$ method again to pull all the virtual $Z$ terms to the end of the circuit by modifying the new Hamiltonians to their double primed counterparts and potentially employing a new time-dilation: $\tau'=f'(\tau)$ such that $\tau_1',\tau_2'=f'^{-1}(\tau_1,\tau_2)$.}
    \label{fig: three qubit example}
\end{figure*}
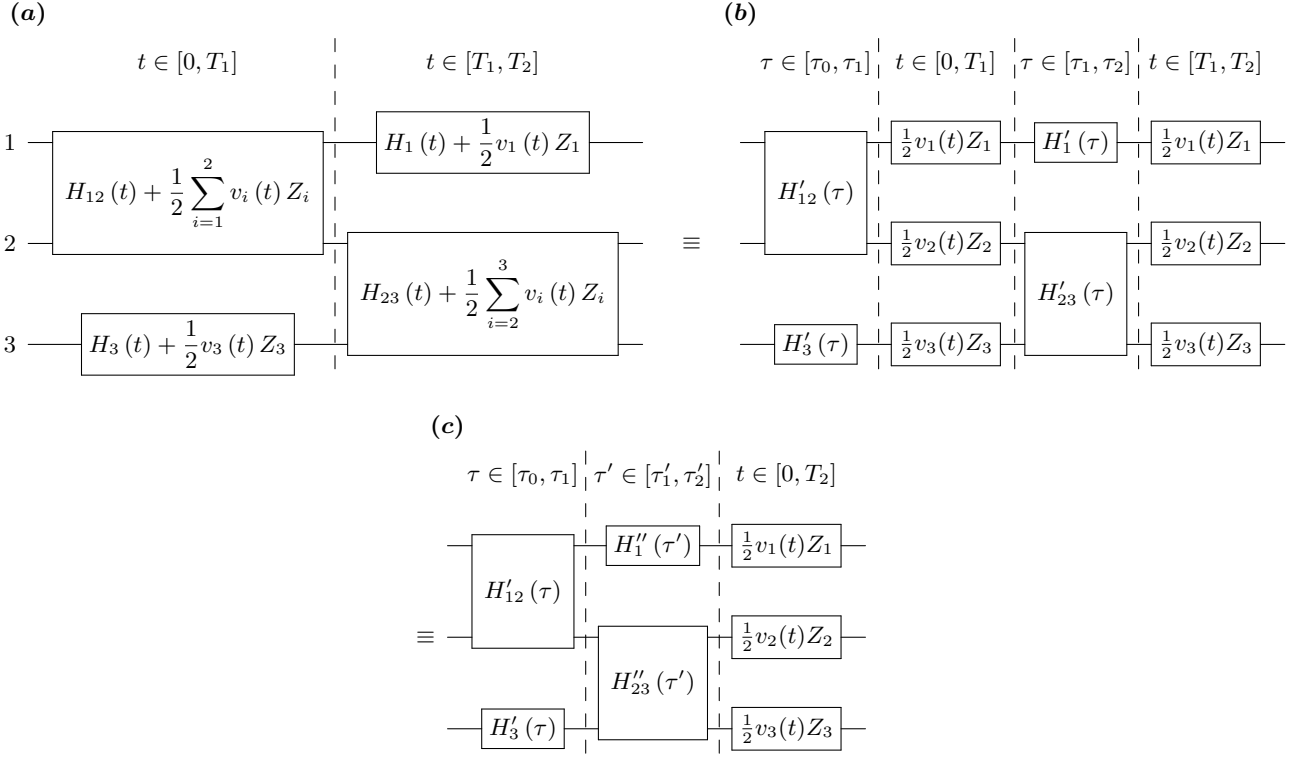

\subsubsection{Time dilation}
However, it may be that we cannot physically implement $\tilde H(t)$. Nonetheless, we may still be able to employ a time dilation, $f$, to implement $\tilde U\left(T\right)$. Let
\begin{equation}\label{eq: def f}
f\colon\tauRange\coloneqq\left[\tau_0,\tau_1\right]\subset\mathbb R\to f\left(\mathcal T\right)\subset\mathbb R;\quad f\colon\tau\mapsto t
\end{equation}
such that $f(\tau_0)=0$, $f(\tau_1)=T$, and $f$ and $\dv{f}{\tau}$ are continuous on $\tauRange$. Thus, $\tRange\subseteq f(\tauRange)$.
Substituting \cref{eq: def f} into \cref{eq: simple differential equation}, we find
\begin{equation}\label{eq: time dialated Schrodinger equation}
  \dv{\tau}\left[\tilde U\circ f\left(\tau\right)\right]=-\frac{i}{\hbar}\dv{f}{\tau}\left[R^\dagger HR\tilde U\right]\circ f\left(\tau\right),\quad\forall\tau\in\tauRange,
\end{equation}
where $g\circ h\left(x\right)\coloneqq g\left(h\left(x\right)\right)$ denotes functional composition. To perform the change of variables, we require $f$ and $\dv{f}{\tau}$ to be continuous on the interval $\tauRange$. Thus, if we can physically implement the Hamiltonian
\begin{equation}\label{eq: modified Hamiltonian}
  H'\left(\tau\right)=\dv{f}{\tau}\left[R^\dagger HR\right]\circ f\left(\tau\right),\quad\forall\tau\in\tauRange,
\end{equation}
for some $f$, then we can implement $\tilde U\left(T\right)$ by first evolving the system from a time $\tau=0$ to $\tau=f^{-1}\left(T\right)$ under the Hamiltonian $H'\left(\tau\right)$: see \cref{fig: three qubit example}.

\subsubsection{Rotating frames}
\label{sec: rotating frame}

Importantly, the Hamiltonian that implements the virtual $Z$ pulse is not unique, but depends on the rotating frame within which we apply the time dilation. Specifically, consider transforming into the frame rotating with the unitary $\hat O(t)$ such that
\begin{equation}\label{eq: commutation relations}
    \left[\hat O(t), Z_i\right]=0\quad\forall i\in\qubits,\ t\in\tRange.
\end{equation}
In this rotating frame, we will denote operators as $\hat\bullet$. The system Hamilonian becomes
\begin{equation}
    \hat H(t)=\hat O^\dagger(t)H(t)\hat O(t)-i\hbar\hat O^\dagger(t)\dv{\hat O}{t},\quad\forall t\in\tRange.
\end{equation}
Similarly, the effective Hamiltonian and the effective Hamiltonian in the virtual frame transform as
\begin{align}
    \hat H_{\textrm{eff}}(t)&=\hat O^\dagger(t)H_{\textrm{eff}}(t)\hat O(t)-i\hbar\hat O^\dagger(t)\dv{\hat O}{t}\\
    \textrm{and}\quad\hat{\tilde H}(t)&=\hat O^\dagger(t)\tilde H(t)\hat O(t)-i\hbar\hat O^\dagger(t)\dv{\hat O}{t},
\end{align}
for all $t\in\tRange$, respectively.
However, using the commutation relations in \cref{eq: commutation relations} we find
\begin{align}
    \hat H_{\textrm{eff}}(t)&=\hat H\left(t\right)+\sum_{i=1}^{n}v_i\left(t\right)Z_i\\
    \textrm{and}\quad\hat{\tilde H}(t)&=R^\dagger\left(t\right)\hat H\left(t\right)R\left(t\right),
\end{align}
for all $t\in\tRange$. Thus, \cref{eq: effective Hamiltonian,eq: simple Hamiltonian} are invariant under such a transformation.

Finally, consider transforming $H'\left(t\right)$ into the \textit{dilated rotating frame}, ${\hat O\circ f\left(\tau\right)}$:
\begin{multline}
    \hat H'\left(\tau\right)=\hat O^\dagger\circ f\left(\tau\right) H'\left(\tau\right)\hat O\circ f\left(\tau\right)\\-i\hbar\hat O^\dagger\circ f\left(\tau\right)\dv{f}{\tau}\left.\dv{\hat O}{t}\right|_{t=f\left(\tau\right)},\quad\forall\tau\in\tauRange.
\end{multline}
Now, substituting in \cref{eq: modified Hamiltonian} we find
\begin{equation}
    \hat H'\left(\tau\right)=\dv{f}{\tau}\left[R^\dagger\hat HR\right]\circ f\left(\tau\right),\quad\forall\tau\in\tauRange.
\end{equation}
Therefore, the Hamiltonian we physically implement in the rotating frame [\cref{eq: modified Hamiltonian}] is also invariant under such a transformation. Transforming back to the laboratory frame, we find
\vbox{\begin{multline}
    H'_{\hat O}(\tau)\coloneqq \hat O\left(\tau\right) \hat H'\left(\tau\right)\hat O^\dagger\left(\tau\right)\\-i\hbar\hat O\left(\tau\right)\left.\dv{\hat O^\dagger}{t}\right|_{t=\tau},\quad\forall\tau\in\tauRange.
\end{multline}}
Thus, for any dilation $f(\tau)\ne\tau$, the Hamiltonian we utilise to implement a virtual $Z$ pulse depends on the rotating frame within which it was solved for.

\subsection{Application to 2-local Hamiltonians}
\label{sec: local Hamiltonian}

In this section, we solve for the Hamiltonian we physically implement, $H'\left(\tau\right)$, to perform virtual $Z$ pulses, under the assumption that our system Hamiltonian is 2-local. 2-local Hamiltonians model a wide array of quantum computing platforms, including semiconductor spin qubits \cite{RevModPhys.95.025003} and superconducting qubits \cite{10.1063/1.5089550}. Explicit examples are outlined in section \cref{sec: examples}. More specifically, in the frame rotating with the unitary
\begin{equation}\label{eq: K def}
    K\left(t\right)\coloneqq\exp(-i\frac{1}{2}\sum_{i=1}^n\phi_i\left(t\right)Z_i),
\end{equation}
we will assume our system has a Hamiltonian of the form
\begin{equation}\label{eq: general Hamilonian}
  H\left(t\right)=\frac{1}{2}\sum_{i=1}^n\vec r_i\left(t\right)\cdot\vec\sigma_i+\sum_{\substack{{i,j=1}\\{:i< j}}}^nJ_{ij}\left(t\right)\mathcal E_{ij}\left(t\right)\quad\forall t\in\tRange,
\end{equation}
where $\vec r_i\left(t\right)$ modulates the Pauli vector $\vec\sigma_i\coloneqq\left(X_i,Y_i,Z_i\right)$ acting on the $i$th qubit. Further, $J_{ij}\left(t\right)$ modulates the interaction Hamiltonian $\mathcal E_{ij}\left(t\right)$ between the $i$th and $j$th qubits in the rotating frame:
\begin{equation}
  \mathcal E_{ij}\left(t\right)\coloneqq K^\dagger\left(t\right)\mathcal E_{ij}\left(0\right)K\left(t\right).
\end{equation}

Our task now is to substitute \cref{eq: general Hamilonian} into \cref{eq: modified Hamiltonian} and re-express $H'\left(\tau\right)$ in the same form as \cref{eq: general Hamilonian}:
\begin{multline}
    H'\left(\tau\right)=\frac{1}{2}\sum_{i=1}^n\pvec r_i'\left(\tau\right)\cdot\vec\sigma_i+\sum_{\substack{{i,j=1}\\{:i< j}}}^nJ_{ij}'\left(\tau\right)\mathcal E_{ij}\left(\tau\right)\quad\forall\tau\in\tauRange.
\end{multline}
Note the subtle difference between the coupling $\mathcal E_{ij}\left(\tau\right)$ (the function $\mathcal E_{ij}$ evaluated at $\tau$) and the dialated coupling $\mathcal E_{ij}\circ f\left(\tau\right)$ [the function $\mathcal E_{ij}$ evaluated at $f\left(\tau\right)$].

First, we will consider the single-qubit term. Using the identities:
\begin{align}
  e^{i\theta Z}X e^{-i\theta Z}&\equiv\cos(2\theta)X-\sin(2\theta)Y,\\
  e^{i\theta Z}Y e^{-i\theta Z}&\equiv\cos(2\theta)Y+\sin(2\theta)X,
\end{align}
we find
\begin{multline}
    \pvec r'_i\left(\tau\right)=\dv{f}{\tau}\!\begin{bmatrix}
    \cos(\frac{V_i\circ f\left(\tau\right)}{\hbar})&-\sin(\frac{V_i\circ f\left(\tau\right)}{\hbar})&0\\
    \sin(\frac{V_i\circ f\left(\tau\right)}{\hbar})&\cos(\frac{V_i\circ f\left(\tau\right)}{\hbar})&0\\
    0&0&1
    \end{bmatrix}\!\vec r_i\circ f\left(\tau\right)\\\forall i\in\qubits,\ \tau\in\tauRange.
\end{multline}
At this point, $f\left(\tau\right)$ is only constrained by any system-specific constraints on $\pvec r_i'\left(\tau\right)$. However, we will find that a non-zero $J_{ij}\left(t\right)$ will in general fix $f\left(\tau\right)$.

Second, we will consider the interaction term for a single pair of qubits $i,j\in\qubits$ such that $i\ne j$. We need to find solutions for $J'_{ij}$ and $f$ to the equation
\begin{align}
    \dv{f}{\tau}\left[J_{ij}R^\dagger\mathcal E_{ij}R\right]\circ f\left(\tau\right)=J'_{ij}\left(\tau\right)\mathcal E_{ij}\left(\tau\right)\quad\forall\tau\in\tauRange.\\
    \nonumber
\end{align}
As $R\circ f\left(\tau\right)$ is unitary, we can factor the equation into two parts, the scalar magnitude and the unitary rotations on $\mathcal E_{ij}\left(0\right)$:
\begin{align}
    J'_{ij}\left(\tau\right)&=\dv{f}{\tau}J_{ij}\circ f\left(\tau\right),\label{eq: J equation}\\\nonumber\\
    \mathcal E_{ij}\left(\tau\right)&=\left[R^\dagger\mathcal E_{ij}R\right]\circ f\left(\tau\right),\label{eq: f equation}
\end{align}
for all $\tau\in\tauRange$, respectively. Therefore, once we solve for $f\left(\tau\right)$ using \cref{eq: f equation} we can subsitute into \cref{eq: J equation} to solve for $J_{ij}'\left(\tau\right)$. To solve \cref{eq: f equation}, we first define, with respect to the Hilbert-Schmidt inner product, the following orthonormal basis of Hermitian operators acting on a four-dimensional Hilbert space
\begin{equation}
    \mathcal B\coloneqq\mathcal Z\cup\mathcal C_+\cup\mathcal C_-\cup\mathcal Q_i\cup\mathcal Q_j,
\end{equation}
divided into the five subspaces
\begin{align}
    \mathcal Z&\coloneqq\left\{\textrm{id.},Z_i,Z_j,Z_iZ_j\right\},\\
    \mathcal C_\pm&\coloneqq\left\{X_iX_j\mp Y_iY_j,X_iY_j\pm Y_jX_i\right\},\\
    \mathcal Q_k&\coloneqq\left\{X_k,Y_k,X_kZ_{\bar k},Y_kZ_{\bar k}\right\},
\end{align}
where $k\in\left\{i,j\right\}$ and $\bar k$ is the complimentary index to $k$. Decomposing $\mathcal E_{ij}\left(t\right)$ in this basis, we find the coefficients,
\begin{equation}
    a_B(t)\coloneqq\frac{1}{4}\Trace(B\mathcal E_{ij}\left(t\right))\quad\forall B\in\mathcal B,
\end{equation}
in each subspace have the following dependencies on $\phi_i\left(t\right)$ and $\phi_j\left(t\right)$:
\begin{align}
    a_z(t)&\textrm{ is a constant}&&\forall z\in\mathcal Z,\\
    a_c(t)&\textrm{ is a function of }\phi_i\left(t\right)\pm\phi_j\left(t\right)&&\forall c\in\mathcal C_{\pm},\\
    a_q(t)&\textrm{ is a function of }\phi_k&&\forall q\in\mathcal Q_k.
\end{align}
Thus, depending on the support of $\mathcal E_{ij}\left(0\right)$, \cref{eq: f equation} simplifies to the following scalar equations
\begin{widetext}
\begin{subnumcases}{\label{eq: f equations}}
    \textrm{all $f\left(\tau\right)$ are solutions}&$\mathcal E_{ij}\left(0\right)\in\operatorname{span}_{\mathbb R}\mathcal Z$,\\
    \Delta^{\pm}\phi_{ij}\left(\tau\right)=\Delta^{\pm}\phi_{ij}\circ f\left(\tau\right)+\tfrac{1}{\hbar}\Delta^{\pm}V_{ij}\circ f\left(\tau\right)\mod 4\pi&$\mathcal E_{ij}\left(0\right)\in\operatorname{span}_{\mathbb R}\mathcal Z\cup\mathcal C_{\pm}$,\\
    \phi_k\left(\tau\right)=\phi_k\circ f\left(\tau\right)+\tfrac{1}{\hbar}V_k\circ f\left(\tau\right)\mod 4\pi&$\mathcal E_{ij}\left(0\right)\in\operatorname{span}_{\mathbb R}\mathcal Z\cup\mathcal Q_k$,\\
    \begin{cases}
        \phi_i\left(\tau\right)=\phi_i\circ f\left(\tau\right)+\tfrac{1}{\hbar}V_i\circ f\left(\tau\right)\mod 4\pi\\
        \phi_j\left(\tau\right)=\phi_j\circ f\left(\tau\right)+\tfrac{1}{\hbar}V_j\circ f\left(\tau\right)\mod 4\pi
    \end{cases}&otherwise,
\end{subnumcases}
for all $\tau\in\tauRange$, where $\Delta^{\pm}\phi_{ij}\left(\tau\right)\coloneqq\phi_i\left(\tau\right)\pm\phi_j\left(\tau\right)$ and $\Delta^{\pm}V_{ij}\left(\tau\right)\coloneqq V_i\left(\tau\right)\pm V_j\left(\tau\right)$. These equations then admit the following solutions:
\begin{subnumcases}{f\left(\tau\right)=\label{eq: f solution}}
    \textrm{any function}&$\mathcal E_{ij}\left(0\right)\in\operatorname{span}_{\mathbb R}\mathcal Z$,\\
    \left(\Delta^{\pm}\phi_{ij}+\tfrac{1}{\hbar}\Delta^{\pm}V_{ij}+4m\pi\right)^{-1}\circ\Delta^{\pm}\phi_{ij}\left(\tau\right)&$\mathcal E_{ij}\left(0\right)\in\operatorname{span}_{\mathbb R}\mathcal Z\cup\mathcal C_{\pm}$,\\
    \left(\phi_k+\tfrac{1}{\hbar}V_k+4m\pi\right)^{-1}\circ\phi_k\left(\tau\right)&$\mathcal E_{ij}\left(0\right)\in\operatorname{span}_{\mathbb R}\mathcal Z\cup\mathcal Q_k$,\\
    \left(\phi_i+\tfrac{1}{\hbar}V_i+4m\pi\right)^{-1}\circ\phi_i\left(\tau\right)&$\begin{multlined}\mathcal E_{ij}\left(0\right)\not\in\left[\operatorname{span}_{\mathbb R}\mathcal Z\cup\mathcal C_\pm\right]\cup\left[\operatorname{span}_{\mathbb R}\mathcal Z\cup\mathcal Q_k\right]\\\textrm{and }\left(\phi_i+\tfrac{1}{\hbar}V_i\right)^{-1}\circ\phi_i\left(\tau\right)=\left(\phi_j+\tfrac{1}{\hbar}V_j\right)^{-1}\circ\phi_j\left(\tau\right),\end{multlined}$\label{eq: special case}\\
    \textrm{no solution}&otherwise,
\end{subnumcases}
for all $\tau\in\tauRange$ and any $m\in\mathbb Z$, where $g^{-1}$ is the inverse function of the function $g$.

Further, it will be useful later to have a general expression for $\dv{f}{\tau}$ in terms of $f\left(\tau\right)$. We can find this by differentiating \cref{eq: f equations}:
\begin{subnumcases}{\dv{f}{\tau}=\label{eq: derivative}}
    \textrm{no general form in terms of $\phi_k$ and $V_k$}&$\mathcal E_{ij}\left(0\right)\in\operatorname{span}_{\mathbb R}\mathcal Z$,\\
    \frac{\Delta^{\pm}\dot\phi_{ij}\left(\tau\right)}{\Delta^{\pm}\dot\phi_{ij}\circ f\left(\tau\right)+\tfrac{1}{\hbar}\Delta^{\pm}v_{ij}\circ f\left(\tau\right)}&$\mathcal E_{ij}\left(0\right)\in\operatorname{span}_{\mathbb R}\mathcal Z\cup\mathcal C_{\pm}$,\\
    \frac{\dot\phi_k\left(\tau\right)}{\dot\phi_k\circ f\left(\tau\right)+\tfrac{1}{\hbar}v_k\circ f\left(\tau\right)}&$\mathcal E_{ij}\left(0\right)\in\operatorname{span}_{\mathbb R}\mathcal Z\cup\mathcal Q_k$,\\
    \frac{\dot\phi_i\left(\tau\right)}{\dot\phi_i\circ f\left(\tau\right)+\tfrac{1}{\hbar}v_i\circ f\left(\tau\right)}&$\begin{multlined}\mathcal E_{ij}\left(0\right)\not\in\left[\operatorname{span}_{\mathbb R}\mathcal Z\cup\mathcal C_\pm\right]\cup\left[\operatorname{span}_{\mathbb R}\mathcal Z\cup\mathcal Q_k\right]\\\textrm{and }\left(\phi_i+\tfrac{1}{\hbar}V_i\right)^{-1}\circ\phi_i\left(\tau\right)=\left(\phi_j+\tfrac{1}{\hbar}V_j\right)^{-1}\circ\phi_j\left(\tau\right),\end{multlined}$\\
    \textrm{no solution}&otherwise,
\end{subnumcases}
for all $\tau\in\tauRange$.
\end{widetext}
If the inverse function in \cref{eq: f solution} is a one-to-many map, then a branch must be picked so that it is a function. Note that for $\mathcal E_{ij}\left(0\right)\in\left[\operatorname{span}_{\mathbb R}\mathcal Z\cup\mathcal C_\pm\right]\cup\left[\operatorname{span}_{\mathbb R}\mathcal Z\cup\mathcal Q_k\right]$ there always exists a solution for $f\left(\tau\right)$. However, for $\mathcal E_{ij}\left(0\right)\not\in\left[\operatorname{span}_{\mathbb R}\mathcal Z\cup\mathcal C_\pm\right]\cup\left[\operatorname{span}_{\mathbb R}\mathcal Z\cup\mathcal Q_k\right]$ only specific $V_i\left(t\right)$ and $V_j\left(t\right)$ admit solutions for $f\left(\tau\right)$.

\cref{eq: f solution}(a) imples that when $\mathcal E_{ij}\left(0\right)\in\mathcal Z$ no time dilation is required. Thus, virtual $Z$ pulses can be implemented without changing the evolution time, much like the instantaneous virtual $Z$ gate. However, when $\mathcal E_{ij}\left(0\right)\not\in\mathcal Z$, a time dilation is, in general, required. Physically, the time dilation accelerates or decelerates the evolution such that the accumulation of phases in the modified Hamiltonian, $H'(\tau)$, compensates for the phases arising from the virtual $Z$ pulses. However, in cases (b) and (c) in \cref{eq: f equations,eq: f solution}, the dilation only depends on a linear combination of the virtual $Z$ pulses on each qubit. This linear combination of the virtual $Z$ pulses cannot be implemented without altering the evolution time. This leaves a linearly independent combination of the virtual $Z$ pulses on each qubit that can be varied without changing the evolution time, thus retaining the instantaneous behaviour of virtual $Z$ gates.

Unfortunately, it becomes complicated to extend this to more than two qubits, as every pair of qubits with non-zero $J_{ij}\left(t\right)$ must admit the same solution for $f\left(\tau\right)$. A workaround is to only couple pairs of qubits at once---\textit{i.e.}, if $J_{ij}\left(t\right)$ is non-zero, then $J_{il}\left(t\right)$ should be zero for all $l\ne j$. We can split our evolution into layers such that in each layer, only pairs of qubits are coupled in any given layer---see \cref{fig: three qubit example}a. Then, for each pair of qubits we can have a separate time dilation $f\left(\tau\right)$---see \cref{fig: three qubit example}b. This, however, means the evolution of each pair will now finish at a different time. For those that finished earlier, we set all driving fields to zero until the last pair of qubits has completed its evolution. Any drift phase accumulated in this time, along with  $R(T)$, can be absorbed into the $V_{i0}$ for the next layer. After all pairs have completed their evolution, we can move on to the next layer. In this way, we can pull our virtual $Z$ pulses all the way through to the end of the evolution---see \cref{fig: three qubit example}c.

\begin{figure*}
    \includegraphics{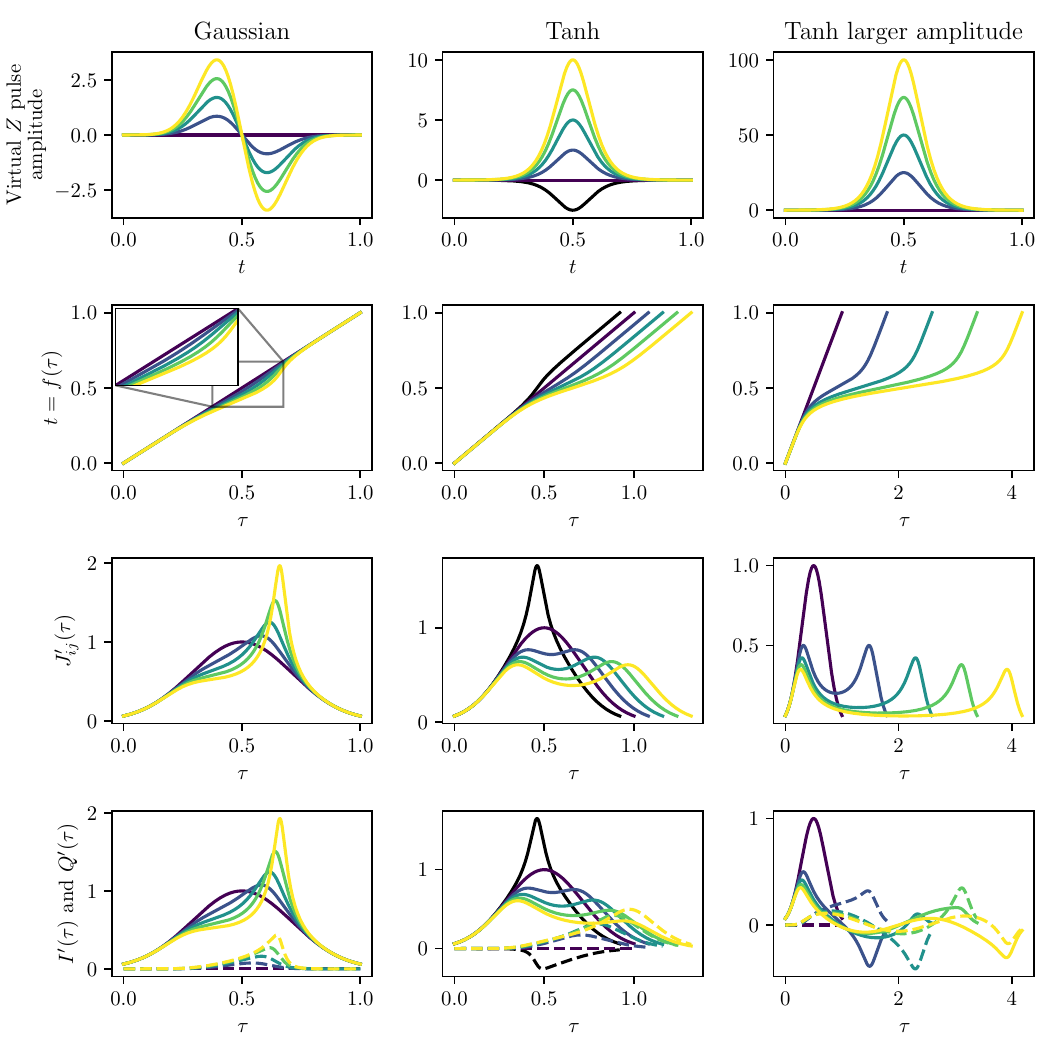}
    \caption{\textbf{Dilated pulses corresponding to virtual $Z$ pulses.} \textbf{(Top row)} The desired virtual $Z$ pulses ($\Delta^\pm_{ij} v(t)$ if $\mathcal E_{ij}\left(0\right)\in\operatorname{span}_{\mathbb R}\mathcal Z\cup\mathcal C_\pm$ or $v_k(t)$ if $\mathcal E_{ij}\left(0\right)\in\operatorname{span}_{\mathbb R}\mathcal Z\cup\mathcal Q_k$), in units with $\hbar=1$, with varying amplitudes: (left) $t\exp(-[(t-0.5)/0.15]^2)$ and (centre and right) $\sech^2([t-0.5]/0.1)$. Thus, the accumulated virtual phases as a function of time are a Gaussian and a hyperbolic tangent, respectively. The centre and right-hand columns differ only by the amplitude, with the right-hand column displaying larger amplitudes. \textbf{(Second row from the top)} The time-dilation $f(\tau)$ [\cref{eq: f solution}] required to implement the virtual $Z$ pulse in the top row. We have taken $\Delta^\pm_{ij} \phi(t)=t$ if $\mathcal E_{ij}\left(0\right)\in\operatorname{span}_{\mathbb R}\mathcal Z\cup\mathcal C_\pm$ or $\phi_k(t)=t$ if $\mathcal E_{ij}\left(0\right)\in\operatorname{span}_{\mathbb R}\mathcal Z\cup\mathcal Q_k$. The black curve in the centre column is an example of a virtual $Z$ pulse that requires time compression. \textbf{(Third row from the top)} The distorted $J_{ij}'(\tau)$ pulse that implements a Gaussian (standard deviation 0.3) $J_{ij}(t)$ pulse in the presence of the desired simultaneous virtual $Z$ pulse. The purple curve corresponds to no virtual $Z$ pulse---\textit{i.e.}, the undistorted pulse. \textbf{(Bottom row)} The distorted $I'_i(\tau)$ (solid) and $Q'_i(\tau)$ (dashed) pulses that implement a Gaussian (standard deviation 0.3) $I_i(t)$ pulse in the presence of the desired simultaneous virtual $Z$ pulse---assuming the virtual $Z$ pulse is applied only to qubit $i$. Note the $J_{ij}(t)$, $I(\tau)$, and $Q(\tau)$ pulses can all be implemented simultaneously.}
    \label{fig: example}
\end{figure*}

\subsubsection{The drift Hamiltonian}
\label{sec: drift}

In \cref{sec: rotating frame} we discovered that the solutions in \cref{sec: local Hamiltonian} will depend on the rotating frame chosen. In this section, we will see that by constraining the single qubit drives, $\vec r_i\left(t\right)$, we are forced to chose the rotating frame of the qubits. If the $z$-component of $\vec r_i\left(t\right)$ [denoted $\hat z\cdot\vec r_i\left(t\right)$] is a control parameter, then we could trivially implement the virtual $Z$ term in the Hamiltonian by absorbing it into $\hat z\cdot\vec r_i\left(t\right)$. The most restrictive constraint would be forcing $\hat z\cdot\vec r_i\left(t\right)$ to be a fixed function of time---\textit{i.e.}, a drift Hamiltonian: potentially time-dependent, but uncontrollable. Thus,
\begin{equation}
    \hat z\cdot\pvec r_i'\left(\tau\right)=\hat z\cdot\vec r_i\left(\tau\right)\quad\forall\tau\in\tauRange.
\end{equation}
However, this precludes any time-dilation $f\left(\tau\right)$ that does not satisfy
\begin{align}
    \hat z\cdot\vec r_i\left(\tau\right)&=\dv{f}{\tau}\hat z\cdot\vec r_i\circ f\left(\tau\right)\quad\forall\tau\in\tauRange\label{eq: f constraint}\\
    \iff h\left(\tau\right)+C&=h\circ f\left(\tau\right)\quad\forall C\in\mathbb R,\ \tau\in\tauRange,\\
    &\textrm{where }h\left(\tau\right)\coloneqq\int_0^\tau\dd{x}\hat z\cdot \vec r_i\left(x\right)\\
    \implies f\left(\tau\right)&=h^{-1}\left(h\left(\tau\right)+C\right)\quad\forall C\in\mathbb R,\ \tau\in\tauRange.\label{eq: f constrained solution}
\end{align}

Thus, by substituting \cref{eq: f constrained solution} into \cref{eq: f equations} and rearanging we find the following constraints

\begin{widetext}
\begin{subnumcases}{\label{eq: v constraints}}
    \textrm{no constraint}&$\mathcal E_{ij}\left(0\right)\in\operatorname{span}_{\mathbb R}\mathcal Z$,\\
    \Delta^{\pm}V_{ij}(t)=\hbar\left[\Delta^{\pm}\phi_{ij}\circ h^{-1}(h(t)-C)-\Delta^{\pm}\phi_{ij}(t)\right]&$\mathcal E_{ij}\left(0\right)\in\operatorname{span}_{\mathbb R}\mathcal Z\cup\mathcal C_{\pm}$,\\
    V_k(t)=\hbar\left[\phi_k\circ h^{-1}(h(t)-C)-\phi_k(t)\right]&$\mathcal E_{ij}\left(0\right)\in\operatorname{span}_{\mathbb R}\mathcal Z\cup\mathcal Q_k$,\\
    \begin{cases}
        V_i(t)=\hbar\left[\phi_i\circ h^{-1}(h(t)-C)-\phi_i(t)\right]\\
        V_j(t)=\hbar\left[\phi_j\circ h^{-1}(h(t)-C)-\phi_j(t)\right]
    \end{cases}&otherwise,
\end{subnumcases}
\end{widetext}
for all $t\in f(\tauRange)$.

If $h$ is injective, then Eqs. (\ref{eq: v constraints}b--d) constrain a linear combination, $l(t)$, of $v_i(t)$ and $v_j(t)$ to a family of functions parameterised by $C$. To allow for more freedom in the virtual $Z$ pulses one can apply, we need $h$ to be non-injective. For each $\tau\in\tauRange$ for which $\hat z\cdot\vec r_i(\tau)=0$, $h^{-1}$ gains an additional branch. If $\hat z\cdot\vec r_i(\tauRange)=0$ for a countable subset of $\tauRange$, then $l(t)$ will be constrained to a union of a countable number of families of functions, each parameterised by one real parameter $C$. To allow for higher-dimensional parameterisations, we require $\hat z\cdot\vec r_i(\tau)=0$ for some finite period of time. This motivates working in the rotating frame of the drift Hamiltonian.

With this theory in place, we present numerically computed pulse distortions for a range of virtual $Z$ pulses in \cref{fig: example}. In \cref{eq: derivative} we can see that as $\Delta^{\pm}v_{ij}\to-\hbar\Delta^{\pm}\dot\phi_{ij}$ if $\mathcal E_{ij}\left(0\right)\in\operatorname{span}_{\mathbb R}\mathcal Z\cup\mathcal C_{\pm}$ or $v_{k}\to-\hbar\dot\phi_{k}$ if $\mathcal E_{ij}\left(0\right)\in\operatorname{span}_{\mathbb R}\mathcal Z\cup\mathcal Q_k$ then $\dv{f}{\tau}\to\pm\infty$. This divergence will cause both the slew rate and pulse amplitude to diverge. In \cref{fig: example}, we observe this behaviour: In the left and centre columns, we see that the more negative the virtual $Z$ pulse, the larger the gradient of the time dilation and the greater the pulse amplitude. Thus, amplitude constraints on a device will constrain how negative a virtual $Z$ pulse one can implement. Interestingly, one can make the virtual $Z$ pulse arbitrarily large and positive without violating device amplitude constraints. Importantly, \cref{fig: example} is hardware agnostic and utilises arbitrary units and so applies to both the semiconducting spin qubit and superconducting qubit examples discussed in the next section.

\section{Example platforms}
\label{sec: examples}

In this section, we consider applying virtual $Z$ pulses on two common quantum computing platforms: semiconductor spin qubits (\cref{sec: semiconductor}) and superconducting qubits (\cref{sec: superconducting}). To allow readers to jump straight to their platform of interest, both \cref{sec: semiconductor,sec: superconducting} are self-contained.

In the following examples, virtual $Z$ pulses can only be applied approximately---under the rotating wave approximation. This will arise in two ways: First, single-qubit drives often do not have enough degrees of freedom to implement virtual $Z$ pulses: we need independently controllable $X$ and $Y$ terms for each qubit in the Hamiltonian. However, the rotating wave approximation allows us to employ frequency multiplexing to obtain sufficient degrees of freedom. Second, by applying the increasingly coarse rotating wave approximations to $\mathcal E_{ij}\left(t\right)$, we will project out $\mathcal C_+$, then, $\mathcal Q_i$ and $\mathcal Q_j$, and finally $\mathcal C_-$ from $\mathcal E_{ij}\left(0\right)$, respectively.

\subsection{Semiconductor spin qubits}
\label{sec: semiconductor}

Semiconductor spin qubits can be modelled with a variety of approaches \cite{RevModPhys.95.025003,PhysRevA.101.022329,PhysRevA.96.052305}. We will focus on the Heisenberg Hamiltonian \cite{RevModPhys.95.025003} in the laboratory frame:
\begin{multline}
    H_{\ce{Si}}\left(t\right)\coloneqq-\frac{\hbar}{2}\sum_{i=1}^n\omega_iZ_i-\sum_{i=1}^n\Omega_i\left(t\right)X_i\\+\frac{1}{4}\sum_{\substack{{i,j=1}\\{:i<j}}}^nJ_{ij}\left(t\right)\left[X_iX_j+Y_iY_j+Z_iZ_j\right].
\end{multline}
Thus, $\mathcal E_{ij}\left(0\right)=X_iX_j+Y_iY_j+Z_iZ_j\in\operatorname{span}_{\mathbb R}\mathcal Z\cup\mathcal C_-$.

It is common to work in the rotating frame of the qubits \cite{RevModPhys.95.025003},
\begin{equation}
    \hat K_{\ce{Si}}(t)\coloneqq\exp(i\frac{1}{2}\sum_{i=1}^n\omega_iZ_it),
\end{equation}
in which the Hamiltonian becomes \cite{RevModPhys.95.025003}
\begin{multline}
    \hat{H}_{\ce{Si}}\left(t\right)=-\sum_{i=1}^n\Omega_i\left(t\right)\left[\cos(\omega_it)X_i+\sin(\omega_it)Y_i\right]\\+\frac{1}{4}\sum_{\substack{{i,j=1}\\{:i<j}}}^nJ_{ij}\left(t\right)\mathcal E_{ij}\left(t\right),
\end{multline}
where $\phi_i\left(t\right)=-\omega_i t$. Unfortunately, $\Omega_i\left(t\right)$ does not provide sufficient degrees of freedom to implement virtual $Z$ pulses, as we need a degree of freedom for each $X_i$ and $Y_i$. However, we can use the rotating wave approximation to implement approximate virtual $Z$ pulses. To achieve this, we consider either a local or a global single-qubit drive with coupling strengths $\mu_i$:
\begin{equation}\label{eq: g def}
    \Omega_i\!\left(t\right)\!\coloneqq\!\begin{cases}
        \!\mu_i\left[I_i\!\left(t\right)\cos(\omega_it)+Q_i\!\left(t\right)\sin(\omega_it)\right]&\textrm{local,}\\
        \displaystyle\!\mu_i\sum_{j=1}^n\left[I_j\!\left(t\right)\cos(\omega_jt)+Q_j\!\left(t\right)\sin(\omega_jt)\right]&\textrm{global,}
    \end{cases}
\end{equation}
for all $i\in\qubits$.

Typically, $I_i\left(t\right)$ and $Q_i\left(t\right)$ are known; if this is not the case then see \cref{app: solution} for formulae for $I_i\left(t\right)$ and $Q_i\left(t\right)$ in terms of $\Omega_i\left(t\right)$. Substituting \cref{eq: g def} into $\hat{H}_{\ce{Si}}$ and ignoring oscillating terms, we find
\begin{multline}
    \hat{H}_{\ce{Si}}\left(t\right)\approx-\frac{1}{2}\sum_{i=1}^n\mu_i\left[I_i\left(t\right)X_i+Q_i\left(t\right)Y_i\right]\\+\frac{1}{4}\sum_{\substack{{i,j=1}\\{:i<j}}}^nJ_{ij}\left(t\right)\mathcal E_{ij}\left(t\right).
\end{multline}
Thus, to apply approximate virtual $Z$ pulses $v_i(t)$ and $v_j(t)$ on a pair of qubits $i$ and $j$ we find
\begin{align}
    &f\left(\tau\right)+\frac{\Delta^-V_{ij}\circ f\left(\tau\right)}{\hbar\left(\omega_i-\omega_j\right)}=\tau,\\
    &\dv{f}{\tau}=\frac{1}{1+\dfrac{\Delta^-v_{ij}\circ f\left(\tau\right)}{\hbar\left(\omega_i-\omega_j\right)}}.
\end{align}
Thus, the new controls are
\begin{align}
    J_{ij}'\left(\tau\right)\!&=\!\frac{J_{ij}\circ f\left(\tau\right)}{1+\dfrac{\Delta^-v_{ij}\circ f\left(\tau\right)}{\hbar\left(\omega_i-\omega_j\right)}},\\
    \begin{bmatrix}
        I_k'\left(\tau\right)\\
        Q_k'\left(\tau\right)
    \end{bmatrix}\!&=\!\frac{\begin{bmatrix}
    \cos(\!\frac{V_k\circ f\left(\tau\right)}{\hbar}\!)&-\sin(\!\frac{V_k\circ f\left(\tau\right)}{\hbar}\!)\\
    \sin(\!\frac{V_k\circ f\left(\tau\right)}{\hbar}\!)&\phantom{-}\cos(\!\frac{V_k\circ f\left(\tau\right)}{\hbar}\!)
    \end{bmatrix}\!\!\!\begin{bmatrix}
        I_k\circ f\left(\tau\right)\\
        Q_k\circ f\left(\tau\right)
    \end{bmatrix}}{1+\dfrac{\Delta^-v_{ij}\circ f\left(\tau\right)}{\hbar\left(\omega_i-\omega_j\right)}}.
\end{align}

Finally, we numerically quantify the magnitude of the error introduced by the rotating wave approximation. We rescale the pulses in \cref{fig: example} such that the qubit frequencies are $18$ and $18.03$ GHz, giving a pulse duration of $33.\dot3$ ns. We also rescale $J$ by $10$ MHz and $I$ and $Q$ by $2.6$ MHz. These values are roughly inline with those presented in Ref. \cite{doi:10.1126/sciadv.abn5130}. The infidelities for the pulse shapes in the left column and the remaining two columns are $4.7\times10^{-7}$ and $1.1\times10^{-5}$, respectively---see \cref{app: infidelities} for more details.

\subsection{Superconducting qubits}
\label{sec: superconducting}

A wide array of superconducting qubits exists \cite{10.1063/1.5089550}. Herein, we limit our considerations to superconducting qubits with tunable couplers (\cref{sec: tunable couplers}), tunable flux (\cref{sec: tunable flux}), and cross-resonance coupling (\cref{sec: cr}).
\subsubsection{Superconducting qubits with tunable couplers}
\label{sec: tunable couplers}

First, we consider the simplest case of superconducting qubits with tunable couplers \cite{10.1063/1.5089550}. The Hamiltonian in the laboratory frame, under the two-level approximation, is \cite{10.1063/1.5089550}
\begin{multline}
    H_{\textrm{tc}}\left(t\right)\coloneqq\frac{\hbar}{2}\sum_{i=1}^n\omega_iZ_i+\sum_{i=1}^n\Omega_i\left(t\right)Y_i\\+\sum_{\substack{{i,j=1}\\{:i< j}}}^nJ_{ij}\left(t\right)\mathcal E_{ij}\left(0\right).
\end{multline}
On the one hand, if the qubits are inductively coupled, then the interaction will be longitudinal---\textit{i.e.}, ${\mathcal E_{ij}\left(0\right)\in\operatorname{span}_{\mathbb R}\mathcal Z}$ \cite{10.1063/1.5089550}. On the other hand, if the qubits are capacitively coupled, the interaction will be transversal---\textit{i.e.}, $\mathcal E_{ij}\left(0\right)\in\operatorname{span}_{\mathbb R}\mathcal C_{+}\cup\mathcal C_{-}$ \cite{10.1063/1.5089550}.

Similar to semiconductor spin qubits \cite{RevModPhys.95.025003}, it is common to work in the rotating frame of the qubits \cite{10.1063/1.5089550},
\begin{equation}
    \hat K_{\textrm{tc}}\coloneqq\exp(-i\frac{1}{2}\sum_{i=1}^n\omega_iZ_it),
\end{equation}
yielding the Hamiltonian
\begin{multline}
    \hat{H}_{\textrm{tc}}\left(t\right)=\sum_{i=1}^n\Omega_i\left(t\right)\left[\cos(\omega_it)Y_i+\sin(\omega_it)X_i\right]\\+\sum_{\substack{{i,j=1}\\{:i< j}}}^nJ_{ij}\left(t\right)\mathcal E_{ij}\left(t\right),
\end{multline}
where $\phi_i\left(t\right)=\omega_it$. As with semiconductor spin qubits, $\Omega_i\left(t\right)$ does not provide sufficient degrees of freedom to implement virtual $Z$ pulses, as we need a degree of freedom for each $X_i$ and $Y_i$. As before, we can use the rotating wave approximation to implement approximate virtual $Z$ pulses. To achieve this, we consider either a local or a global single qubit drive with coupling strengths $\mu_i$:
\begin{equation}\label{eq: superconducting drive}
    \Omega_i\!\left(t\right)\!\coloneqq\!\begin{cases}
        \!\mu_i\left[I_i\!\left(t\right)\cos(\omega_it)+Q_i\!\left(t\right)\sin(\omega_it)\right]&\textrm{local,}\\
        \displaystyle\!\mu_i\sum_{j=1}^n\left[I_j\!\left(t\right)\cos(\omega_jt)+Q_j\!\left(t\right)\sin(\omega_jt)\right]&\textrm{global,}
    \end{cases}
\end{equation}
for all $i\in\qubits$. Taking the rotating wave approximation, we obtain the approximate Hamiltonian \cite{10.1063/1.5089550}
\begin{multline}
    \hat{H}_{\textrm{tc}}\left(t\right)\approx\frac{1}{2}\sum_{i=1}^n\mu_i\left[I_i\left(t\right)Y_i+Q_i\left(t\right)X_i\right]\\+\sum_{\substack{{i,j=1}\\{:i< j}}}^nJ_{ij}\left(t\right)\tilde{\mathcal E}_{ij}\left(t\right),
\end{multline}
where
\begin{equation}
    \tilde{\mathcal E}_{ij}\left(t\right)\coloneqq e^{i\sum_i\phi_i\left(t\right)Z_i}\tilde{\mathcal E}_{ij}\left(0\right)e^{-i\sum_i\phi_i\left(t\right)Z_i},
\end{equation}
and $\tilde{\mathcal E}_{ij}\left(0\right)$ is the projection of $\mathcal E_{ij}\left(0\right)$ onto the $\operatorname{span}_{\mathbb R}\mathcal Z\cup\mathcal C_{-}$ subspace---\textit{i.e.}, $\tilde{\mathcal E}_{ij}\left(0\right)=\mathcal E_{ij}\left(0\right)$ for longitudinal coupling.

Thus, for longitudinal coupling, we can implement virtual $Z$ pulses using the controls
\begin{align}
    J_{ij}'\left(\tau\right)\!&=\!J_{ij}\left(\tau\right),\\
    \begin{bmatrix}
        I_k'\left(\tau\right)\\
        Q_k'\left(\tau\right)
    \end{bmatrix}\!&=\!\begin{bmatrix}
    \cos(\!\frac{V_k\left(\tau\right)}{\hbar}\!)&-\sin(\!\frac{V_k\left(\tau\right)}{\hbar}\!)\\
    \sin(\!\frac{V_k\left(\tau\right)}{\hbar}\!)&\phantom{-}\cos(\!\frac{V_k\left(\tau\right)}{\hbar}\!)
    \end{bmatrix}\!\!\!\begin{bmatrix}
        I_k\left(\tau\right)\\
        Q_k\left(\tau\right)
    \end{bmatrix},
\end{align}
where we did not need a time dilation, so $t=\tau$.

However, for transversal couplings, we must use
\begin{align}
    J_{ij}'\left(\tau\right)\!&=\!\frac{J_{ij}\circ f\left(\tau\right)}{\displaystyle 1+\frac{\Delta^-v_{ij}\circ f\left(\tau\right)}{\hbar\left(\omega_j-\omega_i\right)}},\\
    \begin{bmatrix}
        I_k'\left(\tau\right)\\
        Q_k'\left(\tau\right)
    \end{bmatrix}\!&=\!\frac{\begin{bmatrix}
    \cos(\!\frac{V_k\circ f\left(\tau\right)}{\hbar}\!)&-\sin(\!\frac{V_k\circ f\left(\tau\right)}{\hbar}\!)\\
    \sin(\!\frac{V_k\circ f\left(\tau\right)}{\hbar}\!)&\phantom{-}\cos(\!\frac{V_k\circ f\left(\tau\right)}{\hbar}\!)
    \end{bmatrix}\!\!\!\begin{bmatrix}
        I_k\circ f\left(\tau\right)\\
        Q_k\circ f\left(\tau\right)
    \end{bmatrix}}{\displaystyle 1+\dfrac{\Delta^-v_{ij}\circ f\left(\tau\right)}{\hbar\left(\omega_j-\omega_i\right)}},
\end{align}
where $f\left(\tau\right)$ is the solution to
\begin{equation}
    f\left(\tau\right)+\frac{\Delta^-V_{ij}\circ f\left(\tau\right)}{\hbar\left(\omega_j-\omega_i\right)}=\tau.
\end{equation}

\subsubsection{Flux-tunable superconducting qubits}
\label{sec: tunable flux}

For flux-tunable superconducting qubits, the qubit frequencies depend on the flux applied to the qubit \cite{10.1063/1.5089550}. Additionally, in the absence of tunable couplers, the coupling strength $J_{ij}=J_{ij}\left(\Phi_i\left(t\right),\Phi_j\left(t\right)\right)$ becomes a function of time only through the fluxes $\Phi_i\left(t\right)$ and $\Phi_j\left(t\right)$ applied to qubits $i$ and $j$ \cite{10.1063/1.5089550}. Thus, the Hamiltonian in the laboratory frame becomes \cite{10.1063/1.5089550}
\begin{multline}
        H_{\textrm{ft}}\left(t\right)\coloneqq\frac{\hbar}{2}\sum_{i=1}^n\tilde\omega_i\circ\Phi_i\left(t\right)Z_i+\sum_{i=1}^n\Omega_i\left(t\right)Y_i\\+\sum_{\substack{{i,j=1}\\{:i< j}}}^nJ_{ij}\left(\Phi_i\left(t\right),\Phi_j\left(t\right)\right)\mathcal E_{ij}\left(0\right)
\end{multline}
As before, our couplings can be longitudinal or transversal \cite{10.1063/1.5089550}. We can, of course, implement a $Z$ pulse by varying the flux. However, we may have a constraint on the amount by which the qubit frequency can be varied. Further, we may wish to decouple our $Z$ pulse from our two-qubit coupling strength. This can be solved by using virtual $Z$ pulses.

\textit{Longitudinal coupling:---}As $\mathcal E_{ij}(0)\in\operatorname{span}_{\mathbb R}\mathcal Z$ \cite{10.1063/1.5089550}, we do not require any time dilation to implement virtual $Z$ pulses. Thus, we are free to choose any rotating frame of the form of \cref{eq: K def} to derive the new pulses. To make progress, we generalise our definition of $\Omega_i(t)$ to
\begin{equation}\label{eq: tilde superconducting drive}
    \tilde\Omega_i\!\left(t\right)\!\coloneqq\!\begin{cases}
        \!\!\mu_i\left[I_i\!\left(t\right)\cos\phi_i(t)\!+\!Q_i\!\left(t\right)\sin\phi_i(t)\right]&\textrm{local,}\\
        \displaystyle\!\!\mu_i\sum_{j=1}^n\left[I_j\!\left(t\right)\cos\phi_j(t)\!+\!Q_j\!\left(t\right)\sin\phi_j(t)\right]&\textrm{global,}
    \end{cases}
\end{equation}
for all $i\in\qubits$. Substituting for and taking the rotating wave approximation, we find
\vbox{\begin{multline}\label{eq: ft Hamiltonian}
    \hat{H}_{\textrm{ft}}\approx\frac{\hbar}{2}\sum_{i=1}^n\left[\tilde\omega_i\circ\Phi_i\left(t\right)+\phi_i(t)\right]Z_i\\+\frac{1}{2}\sum_{i=1}^n\mu_i\left[I_i\left(t\right)X_i+Q_i\left(t\right)Y_i\right]\\+\sum_{\substack{{i,j=1}\\{:i< j}}}^nJ_{ij}\left(\Phi_i\left(t\right),\Phi_j\left(t\right)\right)\mathcal E_{ij}\left(t\right).
\end{multline}}
Note that the frequencies dropped in the rotating wave approximation will depend on the function $\phi_i(t)$ chosen. Nonetheless, we can implement virtual $Z$ pulses with the controls
\begin{align}
    \begin{bmatrix}
        I_k'\left(\tau\right)\\
        Q_k'\left(\tau\right)
    \end{bmatrix}\!&=\!\begin{bmatrix}
    \cos(\!\frac{V_k\left(\tau\right)}{\hbar}\!)&-\sin(\!\frac{V_k\left(\tau\right)}{\hbar}\!)\\
    \sin(\!\frac{V_k\left(\tau\right)}{\hbar}\!)&\phantom{-}\cos(\!\frac{V_k\left(\tau\right)}{\hbar}\!)
    \end{bmatrix}\!\!\!\begin{bmatrix}
        I_k\left(\tau\right)\\
        Q_k\left(\tau\right)
    \end{bmatrix},\\
    \Phi_i'\left(\tau\right)&=\Phi_i\left(\tau\right).
\end{align}

\textit{Transversal coupling:---}Now $\mathcal E_{ij}(0)\not\in\operatorname{span}_{\mathbb R}\mathcal Z$ \cite{10.1063/1.5089550}, so we will need to employ a time dilation. A dilation will require us to rescale $J_{ij}$ and so update the fluxes $\Phi_k$. However, updating the fluxes will change the qubit frequencies. Until now, we have worked in the rotating frame of the qubit frequencies to remove the need to rescale the qubit frequencies---see \cref{sec: drift}. Thus, we cannot work within the same rotating frame. However, this presents a good opportunity to explore an interesting property of virtual $Z$ pulses identified in \cref{sec: rotating frame}: the solution depends on the rotating frame within which it is solved.

To obtain enough degrees of freedom to implement virtual $Z$ pulses, we also solve for the rotating frame, that admits a solution. That is, we need to solve the following system of equations:
\begin{subnumcases}{\label{eq: ft constraints}}
    \text{\cref{eq: f equations}},\\
    \begin{aligned}
        &\tilde\omega_k\circ\Phi_k'(\tau)-\dot\phi_k(\tau)\\&=\!\dv{f}{\tau}\!\!\left[\tilde\omega_k\!\circ\!\Phi_k\!\circ\! f(\tau)\!-\!\dot\phi_k\circ f(\tau)\right]\forall k\in\left\{i, j\right\},
    \end{aligned}\\
    J_{ij}\left(\Phi_i'\!\left(\tau\right),\Phi_j'\!\left(\tau\right)\right)\!=\!\dv{f}{\tau}J_{ij}\left(\Phi_i\!\circ\! f\!\left(\tau\right),\Phi_j\!\circ\! f\!\left(\tau\right)\right),
\end{subnumcases}
for $f(\tau)$, $\phi_i(\tau)$, $\phi_j(\tau)$, $\Phi_i(\tau)$, and $\Phi_j(\tau)$. We have between three and five constraints, depending on the case in \cref{eq: f equations}, and five variables. Thus, the variables may no longer be overconstrained.

\begin{figure}
    \includegraphics{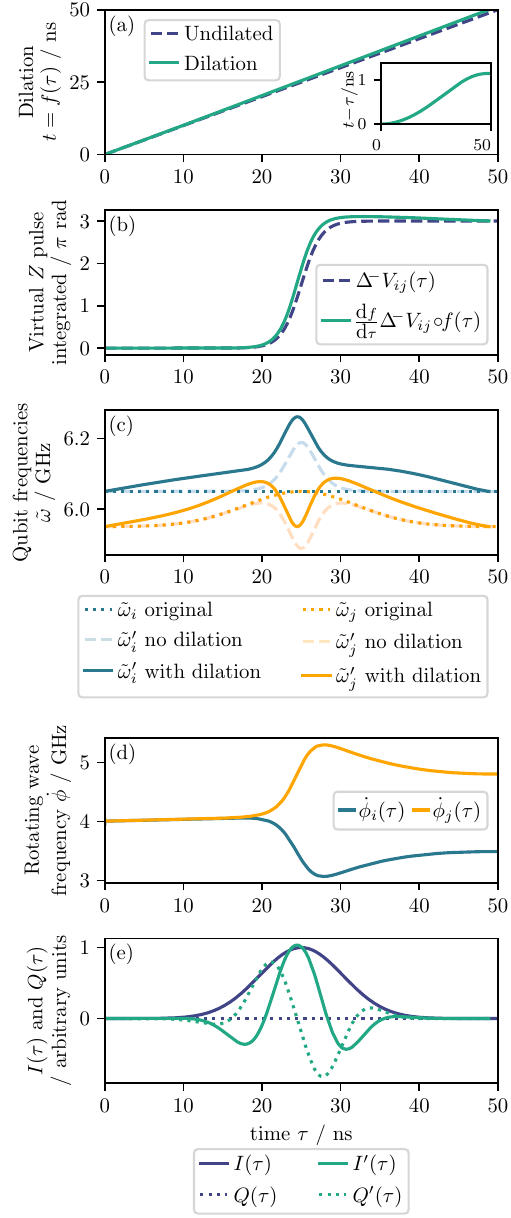}
    \vspace{-2.1em}
    \caption{\textbf{An example virtual $Z$ pulse on flux-tunable superconducting qubits.} \textbf{(a)} The optimised time dilation (solid) corresponding to the virtual $Z$ pulse in \textbf{(b)}. \textbf{(c)} The time dependence of the original qubit frequencies (dotted curve) and the new qubit frequencies with(or without) time dilation [dashed (or solid)]. \textbf{(d)} The frequency of the rotating frame used for each qubit. \textbf{(e)} The distorted $I'_i(\tau)$ and $Q'_i(\tau)$ pulses corresponding to the original $I_i(\tau)$ and $Q_i(\tau)$ pulses.}
    \label{fig: ft sc example}
\end{figure}

As an example, we take $f(\tau)$ to be a free function and solve for the other four functions as outlined in \cref{app: flux solution}. We optimise $f(\tau)$ to minimise the pulse distortion required to implement the virtual $Z$ pulse---see \cref{app: dilation optimisation} for details. We solve \cref{eq: ft constraints} numerically with $\mathcal E_{ij}(0)\in\operatorname{span}_{\mathbb R}\mathcal Z\cup \mathcal C_{-}$, $J_{ij}\propto\frac{1}{\tilde\omega_i-\omega_r}+\frac{1}{\tilde\omega_j-\omega_r}$ (as for capacitive couplings mediated through a bus resonator \cite{10.1063/1.5089550}) with frequency $\omega_r=8$ GHz, $T=50$ ns, $\Delta^{-}v_{ij}(t)=0.3\sech^2(\frac{t-25}{2.5})$ GHz, and the qubit frequencies $\tilde\omega_i\circ\Phi_i(t)=6.05$ GHz and ${\tilde\omega_j\circ\Phi_j(t)=5.95+0.05\exp[-\left(\frac{t-25}{10}\right)^2]}$. The solutions are presented in \cref{fig: ft sc example}(a--d). In \cref{fig: ft sc example}(c), we see that without the time dilation, the qubits would need to decrease in frequency. Typically, the flux sweet spot is a minimum in the qubit frequency \cite{10.1063/1.5089550}, making this infeasible. Fortunately, an optimised time dilation prevents the frequencies dropping below this minimum value---the frequency at $t=0$.

After solving \cref{eq: ft constraints}, we can utilise the values of $\phi_k(t)$ found to decompose the single-qubit drives as in \cref{eq: tilde superconducting drive} using \cref{app: solution}. Taking the rotating wave approximation, we can find \cref{eq: ft Hamiltonian}. With this approximate Hamiltonian, we can update the single-qubit drives:
\begin{equation}
    \begin{bmatrix}
        I_k'\!\left(\tau\right)\\
        Q_k'\!\left(\tau\right)
    \end{bmatrix}\!=\!\dv{f}{\tau}\!\!\begin{bmatrix}
        \cos(\!\frac{V_k\circ f\left(\tau\right)}{\hbar}\!)&-\sin(\!\frac{V_k\circ f\left(\tau\right)}{\hbar}\!)\\
        \sin(\!\frac{V_k\circ f\left(\tau\right)}{\hbar}\!)&\phantom{-}\cos(\!\frac{V_k\circ f\left(\tau\right)}{\hbar}\!)
        \end{bmatrix}\!\!\!\begin{bmatrix}
            I_k\!\circ\! f\!\left(\tau\right)\\
            Q_k\!\circ\! f\!\left(\tau\right)
        \end{bmatrix}.
\end{equation}
However, $\tilde\Omega_k'\left(\tau\right)$ still has time-dependent carrier frequencies. We can use the formulae in \cref{app: solution} again to convert back to a constant carrier frequency implementation. An example is presented in \cref{fig: ft sc example}(e).

Finally, we quantify the error introduced by the rotating wave approximation. Taking both the coupling strength and $I(\tau)$ pulse amplitude as $10$ MHz \cite{10.1063/1.5089550}, we numerically find that the distorted pulse in \cref{fig: ft sc example} implements the desired evolution with an infidelity of $3.3\times10^{-6}$.

\subsubsection{Cross-resonance superconducting qubits}
\label{sec: cr}

For cross-resonance superconducting qubits \cite{10.1063/1.5089550}, neither the fluxes nor the couplings are tunable. However, by utilising the Schrieffer-Wolff transformation \cite{PhysRevA.102.042605} to decouple the states in the coupling resonators from the computational space to leading order, we obtain the following Hamiltonian \cite{10.1063/1.5089550,PhysRevA.102.042605}:
\begin{equation}
H_{\textrm{cr}}\left(t\right)=-\frac{\hbar}{2}\sum_{i=1}^n\omega_iZ_i+\sum_{i,j=1}^n\Omega_i\left(t\right)\left[\nu_{ij}X_j+\mu_{ij}Z_iX_j\right],
\end{equation}
where $\nu_{ii}=1$ and $\mu_{ii}=0$.
Moving to the rotating frame of the qubits,
\begin{equation}
    \hat K_{\textrm{cr}}\coloneqq\exp(i\frac{1}{2}\sum_{i=1}^n\omega_iZ_it),
\end{equation}
the Hamiltonian becomes
\begin{multline}
    \hat{H}_{\textrm{cr}}\left(t\right)=\sum_{i,j=1}^n\Omega_i\left(t\right)[\nu_{ij}\cos(\omega_jt)X_j-\nu_{ij}\sin(\omega_jt)Y_j\\+\mu_{ij}\cos(\omega_jt)Z_iX_j-\mu_{ij}\sin(\omega_jt)Z_iY_j].
\end{multline}
To yield sufficient degrees of freedom to implement virtual $Z$ pulses, we assume single-qubit drives take the following local form
\begin{equation}
    \Omega_i\left(t\right)=\sum_{j=1}^nI_{ij}\left(t\right)\cos(\omega_jt)+Q_{ij}\left(t\right)\sin(\omega_jt).
\end{equation}
Substituting for $\Omega_i\left(t\right)$ and dropping oscillating terms in the rotating wave approximation, we find the approximate Hamiltonian \cite{PhysRevA.102.042605}
\begin{samepage}
\begin{multline}
    \hat{H}_{\textrm{cr}}\left(t\right)\approx\frac{1}{2}\sum_{i=1}^n[\nu_{ij}I_{ij}\left(t\right)X_j-\nu_{ij}Q_{ij}\left(t\right)Y_j\\+\mu_{ij}I_{ij}\left(t\right)Z_iX_j-\mu_{ij}Q_{ij}\left(t\right)Z_iY_j].
\end{multline}
\end{samepage}
Thus, to implement virtual $Z$ pulses approximately, we use
\begin{equation}
    \begin{bmatrix}
        I_{ij}'\left(\tau\right)\\
        Q_{ij}'\left(\tau\right)
    \end{bmatrix}\!=\!\begin{bmatrix}
    \phantom{-}\cos(\!\frac{V_j\left(\tau\right)}{\hbar}\!)&\sin(\!\frac{V_j\left(\tau\right)}{\hbar}\!)\\
    -\sin(\!\frac{V_j\left(\tau\right)}{\hbar}\!)&\cos(\!\frac{V_j\left(\tau\right)}{\hbar}\!)
    \end{bmatrix}\!\!\!\begin{bmatrix}
        I_{ij}\left(\tau\right)\\
        Q_{ij}\left(\tau\right)
    \end{bmatrix},
\end{equation}
where we did not need a time dilation, so $t=\tau$.

\section{Potential Applications}

We have considered how to implement virtual $Z$ pulses on several platforms. Now we turn to the potential applications of virtual $Z$ pulses. We identify three applications: gate design (\cref{sec: gate design}), variational quantum algorithms (\cref{sec: VQAs}), and Hamiltonian simulation (\cref{sec: Hamiltonian simulation}).

\subsection{Wider hardware support for virtual \texorpdfstring{$Z$}{Z} gates}
\label{sec: gate design}

An existing method for implementing virtual $Z$ gates is using frame tracking \cite{PRXQuantum.5.020338,PhysRevResearch.6.013235}. Frame tracking allows virtual $Z$ gates to be pulled through any multi-qubit gate that normalises the Lie group generated by weight-one Pauli-$Z$ strings, $L\coloneqq\exp(i\operatorname{span}_{\mathbb R}\left\{Z_i\right\}_{i=1}^n)$. These multi-qubit gates form the normaliser group
\begin{equation}
    \mathcal N\coloneqq\left\{u\in\operatorname{U}(2^n):uL=Lu\right\},
\end{equation}
which is generated by $X_1$, $\left\{\text{SWAP}_{i,i+1}\right\}_{i=1}^{n-1}$, and arbitrary phase gates, $\exp(i\operatorname{span}_{\mathbb R}\left\{I,Z\right\}^{\otimes n})$---see \cref{app: normalizer} for a proof. This group includes some notable two-qubit gates which are compatible with virtual $Z$ gates: SWAP, $i$SWAP, $b$SWAP, and CPHASE. However, CNOT and general powers of SWAP, $i$SWAP, and $b$SWAP are not in $\mathcal N$. Thus, many parameterised two-qubit gates do not support the use of virtual $Z$ gates. The two existing solutions are utilising microwave-activated multi-qubit gates \cite{PhysRevLett.129.060501,PRXQuantum.6.010349,PhysRevApplied.20.024011,Wu2024,Kim2022,PhysRevApplied.14.014072,PhysRevA.98.052318,PhysRevResearch.2.033447,PRXQuantum.5.020326}, or implementing physical $Z$ gates either via idling or a combination of $X$ and $Y$ rotations. This work presents a new alternative: By employing virtual $Z$ pulses, we can pull virtual $Z$ gates through pulses that generate CNOT and arbitrary powers of SWAP, $i$SWAP, and $b$SWAP. In fact, we can pull virtual $Z$ pulses through any two-qubit gate generated by one of the Hamiltonian algebras $\operatorname{span}_{\mathbb R}\mathcal Z\cup\mathcal C_{\pm}$ or $\operatorname{span}_{\mathbb R}\mathcal Z\cup\mathcal Q_k$ supplemented with simultaneous driving of $X_i$ and $Y_i$ via $I_i(t)$ and $Q_i(t)$---\textit{i.e.}, any $\operatorname{U}(2)$ gate, as all four Hamiltonian algebras yield the full $\mathfrak u(2)$ dynamic Lie algebra. As the virtual $Z$ gate is not applied simultaneously with the two-qubit gate the pulse distortion is similar to applying physical $Z$ gates. The distortion corresponds to an ideling delay before applying the two-qubit gate. However, the delay is not chosen to implement the whole $Z$ rotation, but to accumulate the phase that will not commute through the two-qubit gate.

\subsection{Variational quantum algorithms}
\label{sec: VQAs}

Variational quantum algorithms (VQAs) \cite{Cerezo2021,TILLY20221,BLEKOS20241} are a promising contender for demonstrating quantum advantage on noisy intermediate-scale quantum devices. VQAs evolve an initial state on a quantum computer under an ansatz specified by a parameterised unitary. Next, VQAs measure a cost function of the output state. A classical optimiser then minimises the cost function.

In general, virtual $Z$ pulses could be employed to increase the expressiveness of pulse-based ans\"atze \cite{doi:10.1126/science.abo6587,PRXQuantum.2.010101,PhysRevApplied.23.024036,Meitei2021,PhysRevResearch.5.033159,9996174,10.3389/frqst.2023.1273581,PhysRevD.111.034506,PhysRevX.7.021027,PhysRevApplied.19.064071,PhysRevLett.118.150503,Lu2017,Long2025} and potentially decrease the evolution time \cite{PhysRevApplied.19.064071,Long2025}. As the optimiser already optimises the pulse shapes, we can absorb the pulse distortions into the optimiser: We only need to prepend and append parameterised $Z$ rotations to the ansatz.

Due to the constraints outlined above, not all virtual $Z$ pulses can be pulled through the remaining pulse sequence. However, we can try the reverse process to pull the virtual $Z$ pulses through to the beginning of the pulse. This, of course, also may not work. Nonetheless, this justifies both appending and prepending parameterised $Z$ rotations to the ansatz.

Unfortunately, two common cases yield either the appended or prepended parameterised $Z$ rotations redundant: (i) If the measurements are in the computational basis (\textit{e.g.}, QAOA \cite{farhi2014quantumapproximateoptimizationalgorithm,BLEKOS20241,PhysRevX.7.021027} for problems such as MaxCut \cite{farhi2014quantumapproximateoptimizationalgorithm,PhysRevResearch.4.033029,PhysRevA.109.032420} or financial modeling \cite{hodson2019portfoliorebalancingexperimentsusing,Brandhofer2022,chen2024quasibinaryencodingbasedquantum,yuan2024quantifyingadvantagesapplyingquantum}) the virtual $Z$ pulses will not increase the expressiveness as they will not affect the measurement outcome. (ii) Similarly, if the input state is a computational basis state (\textit{e.g.}, the Hartree-Fock state within the Jordan-Wigner encoding for molecular simulations \cite{RevModPhys.92.015003}), then virtual $Z$ pulses will not increase the expressiveness.

In many cases, the prepended and appended parameterised $Z$ rotations can be virtualised. The appended $Z$ rotations can be absorbed into the measurements by updating the measurement basis. The prepended $Z$ rotations can be absorbed into the initial state preparation. For example, for QAOA, it is common to take the initial state as $\ket{+}^{\otimes n}$. By absorbing the parameterised $Z$ rotation, we take the initial state to be $\left[\cos(\theta)\ket{+}+i\sin(\theta)\ket{-}\right]^{\otimes n}$. The original initial state can be prepared with single-qubit Hadamard gates; the new initial state can be prepared with the same pulse sequence but a shifted microwave phase.

\subsection{Hamiltonian simulation}
\label{sec: Hamiltonian simulation}

Hamiltonian simulation \cite{PRXQuantum.2.017003} is both the seminal \cite{Feynman1982} and one of the most promising use cases for quantum computers \cite{PRXQuantum.2.017003}. Virtual $Z$ pulses increase the number of Hamiltonians we can perform analogue simulations of. As an example, we will consider semiconductor spin qubits \cite{RevModPhys.95.025003}. In \cref{sec: semiconductor}, we found that in the rotating frame of the qubits,
\begin{equation}
    \hat K_{\ce{Si}}(t)\coloneqq\exp(i\frac{1}{2}\sum_{i=1}^n\omega_iZ_it),
\end{equation}
we could implement the effective Hamiltonian
\begin{multline}
    \hat{H}_{\ce{Si}\textrm{ eff}}\left(t\right)\coloneqq-\frac{1}{2}\sum_{i=1}^n\left[I_i\left(t\right)X_i+Q_i\left(t\right)Y_i+v_i(t)Z_i\right]\\+\frac{1}{4}\sum_{\substack{{i,j=1}\\{:i<j}}}^nJ_{ij}\left(t\right)\mathcal E_{ij}\left(t\right),
\end{multline}
providing that, at each $t\in\left[T\right]$ and for each $i$, there was at most one value of $j$ for which $J_{ij}(t)\ne 0$. Moving to the rotating frame with the average Zeeman splitting,
\begin{equation}
    \hat K_{\ce{Si}}^{\textrm{avg}}(t)\coloneqq\exp(i\frac{1}{2n}\sum_{i,j=1}^n\omega_iZ_jt),
\end{equation}
the Hamiltonian becomes
\begin{multline}
    \hat{\hat{H}}_{\ce{Si}\textrm{ eff}}\left(t\right)=\\-\frac{1}{2}\sum_{i=1}^n[\left(I_i\left(t\right)\cos(\Delta\omega_i t)+Q_i\left(t\right)\sin(\Delta\omega_i t)\right)X_i\\+\left(Q_i\left(t\right)\cos(\Delta\omega_i t)-I_i\left(t\right)\sin(\Delta\omega_i t)\right)Y_i+v_i(t)Z_i]\\+\frac{1}{4}\sum_{\substack{{i,j=1}\\{:i<j}}}^nJ_{ij}\left(t\right)\left[X_iX_j+Y_iY_j+Z_iZ_j\right].
\end{multline}
Next, by choosing
\begin{equation}
    \begin{bmatrix}
        I_i\left(t\right)\\
        Q_i\left(t\right)\\
        v_i\left(t\right)
    \end{bmatrix}
    =
    \begin{bmatrix}
        \cos(\Delta\omega_it)&-\sin(\Delta\omega_it)&0\\
        \sin(\Delta\omega_it)&\phantom{-}\cos(\Delta\omega_it)&0\\
        0&0&1
    \end{bmatrix}
    \vec B_i(t),
\end{equation}

\vbox{
\noindent we can simulate a Heisenberg Hamiltonian with arbitrary magnetic fields:
\begin{multline}
    \hat{\hat{H}}_{\ce{Si}\textrm{ eff}}\left(t\right)=-\frac{1}{2}\sum_{i=1}^n\vec B_i(t)\cdot\vec\sigma_i\\+\frac{1}{4}\sum_{\substack{{i,j=1}\\{:i<j}}}^nJ_{ij}\left(t\right)\left[X_iX_j+Y_iY_j+Z_iZ_j\right].
\end{multline}}
Note that without the virtual $Z$ pulses, we would only be able to natively simulate a Heisenberg Hamiltonian
\linebreak\vspace{0.3em}
with the same Zeeman splittings as the qubits up to a
\linebreak\vspace{0.3em}
global field: $B_i^{(z)}(t)=\hbar\omega_i+B_{\textrm{global}}^{(z)}$. One use of this is preparing the ground states of Heisenberg Hamiltonians with adiabatic transitions \cite{farhi2000quantumcomputationadiabaticevolution}.

\section{Discussion and conclusion}
\label{sec: conclusion}

Within this article, we have presented a theoretical framework for implementing virtual $Z$ pulses. This framework is architecture-independent, and we demonstrate example pulse shapes for semiconductor spin qubits and superconducting qubits. In contrast to virtual $Z$ gates, we find that virtual $Z$ pulses can only be implemented approximately on most platforms. Thus, much like the theory of Rabi-oscillations, which also employs the rotating wave approximation, this framework should act as the first step in designing a pulse sequence before applying pulse-shaping methods \cite{Glaser2015,Mahesh2023,Müller_2022} to correct the errors introduced by the approximations. Alternatively, if virtual $Z$ pulses are utilised as degrees of freedom within a variational quantum algorithm, then the classical optimiser can compensate for the approximation errors.

We find that virtual $Z$ pulses can be applied to more hardware platforms than virtual $Z$ gates. Virtual $Z$ gates require the two-qubit gates to either be microwave-activated or to normalise the single-qubit $Z$ rotations. In contrast, virtual $Z$ pulses have a weaker constraint on the Hamiltonian algebra. This allows more freedom in architecture and gate design. We demonstrate that this increased freedom allows for more Hamiltonians to be natively simulated in an analogue manner.

Further, we found that on many platforms, virtual $Z$ pulses can only be implemented if each qubit is coupled to at most one other qubit simultaneously. The trivial exception is qubits coupled by a $ZZ$ interaction. However, the more interesting exception is cross-resonance superconducting qubits.

Much like their gate-based counterparts, we anticipate that virtual $Z$ pulses will allow for faster and more robust quantum information processing. Additionally, virtual $Z$ pulses open a new zoo of native evolutions with increased flexibility and tunability.

\section*{Data availability}
The data collected for all the figures in the Article are available in Ref. \cite{long_2025_17113741}.

\section*{Code availability}
All numerical calculations in this paper were performed in Python using NumPy \cite{harris2020array}, SciPy \cite{2020SciPy-NMeth}, and PySTE \cite{Long_PySTE_2025,Long2025}. The data and figures can be reproduced using the source code available at \href{https://github.com/Christopher-K-Long/Virtual-Z-gates-at-the-pulse-level-source-code}{https://github.com/Christopher-K-Long/Virtual-Z-gates-at-the-pulse-level-source-code} \cite{long_2025_17123423}.

\section*{Acknowledgements}
C.K.L. would like to thank Kieran Dalton, Djamila Hiller, Normann Mertig, and Maximilian Rimbach-Russ for valuable feedback, discussions, and pointers to relevant literature.

\section*{Author contributions}
C.K.L. designed and executed the project, and wrote the software and manuscript. C.H.W.B. supervised the project.

\bibliography{reference}

\onecolumngrid
\appendix
\renewcommand{\thesubsection}{\Alph{section}.\arabic{subsection}}

\section{Solving for \texorpdfstring{$I$}{I} and \texorpdfstring{$Q$}{Q}}
\label{app: solution}

In this article, we frequently decompose the single-qubit drives into the form:
\begin{equation}
    \tilde\Omega_i\!\left(t\right)\!\coloneqq\!\begin{cases}
        \!\!\mu_i\left[I_i\!\left(t\right)\cos\phi_i(t)\!+\!Q_i\!\left(t\right)\sin\phi_i(t)\right]&\textrm{local,}\\
        \displaystyle\!\!\mu_i\sum_{j=1}^n\left[I_j\!\left(t\right)\cos\phi_j(t)\!+\!Q_j\!\left(t\right)\sin\phi_j(t)\right]&\textrm{global.}
    \end{cases}\tag{\ref{eq: tilde superconducting drive}}
\end{equation}
However, it may be the case that $I_i\left(t\right)$ and $Q_i\left(t\right)$ are unknown. Thus, it is important to be able to invert \cref{eq: tilde superconducting drive}. The simplest solution is
\begin{equation}
    \begin{bmatrix}
        I_i\left(t\right)\\
        Q_i\left(t\right)
    \end{bmatrix}=\frac{\tilde\Omega_i\left(t\right)}{N\mu_i}\begin{bmatrix}
        \cos\phi_j(t)\\
        \sin\phi_j(t)
    \end{bmatrix}\qquad\textrm{ where }N=\begin{cases}
        1&\textrm{local,}\\
        N&\textrm{global.}
    \end{cases}
\end{equation}
However, in the case $\phi_j(t)=\omega_it$ and $\Omega_i=\cos(\omega_it)$ this will yield
\begin{equation}
    \begin{bmatrix}
        I_i\left(t\right)\\
        Q_i\left(t\right)
    \end{bmatrix}=\frac{1}{2\mu_i}\begin{bmatrix}
        1+\cos(2\omega_it)\\
        \sin(2\omega_it)
    \end{bmatrix}.
\end{equation}
Ideally, $I_i(t)$ and $Q_i(t)$ should be as close to baseband pulses as possible. Thus, a more optimal solution would be $I_i(t)=1$ and $Q_i(t)=0$.

To obtain lower frequency solutions, we will solve for $I_i(t)$ and $Q_i(t)$ in frequency space. To gain some intuition, we will start with the local form of $\Omega_i$. First we will re-express $\tilde\Omega_i$ in terms of $\phi_i$: Let $\bar\Omega_i\coloneqq\frac{1}{\mu_i}\Omega_i\circ\phi_i^{-1}$, $\bar I_i\coloneqq I_i\circ\phi_i^{-1}$, and $\bar Q_i\coloneqq Q_i\circ\phi_i^{-1}$. Thus,
\begin{equation}
    \bar\Omega_i(t)=\bar I_i(t)\cos(t)+\bar Q_i(t)\sin(t),
\end{equation}
and taking the Fourier transform (assuming $\phi_i(t)\to\pm\infty$ as $t\to\pm\infty$ or $\Omega_i(t)$ vanishes outside a set of finite measure), we find
\begin{equation}
    2\hat{\bar\Omega}_i(\omega)=\hat{\bar I}_i(\omega-1)-i\hat{\bar Q}_i(\omega-1)+\hat{\bar I}_i(\omega+1)+i\hat{\bar Q}_i(\omega+1),
\end{equation}
where $\hat\bullet$ denotes the Fourier transform of $\bullet$. As $\bar\Omega_i$, $\bar I_i$, and $\bar Q_i$ are real functions, then $\hat{\bar\Omega}_i$ and $\hat{\bar I}_i$ are Hermitian functions, and $-i\hat{\bar Q}$ is an anti-Hermitian function. Thus, as any function can be expressed as a linear combination of a Hermitian and an anti-Hermitian function, then $\hat g(\omega)=\hat{\bar I}_i(\omega-1)-i\hat{\bar Q}_i(\omega-1)$ can be any function. Further, note that
\begin{equation}
    \hat{\bar I}_i(\omega+1)+i\hat{\bar Q}_i(\omega+1)=\hat{\bar I}_i^*(-\omega-1)+i\hat{\bar Q}^*_i(-\omega-1)=\hat g^*(-\omega),
\end{equation}
and so $\hat{\bar\Omega}_i(\omega)$ is just the Hermitian component of $\hat g(\omega)$. Thus, we can express the solutions for $\hat g(\omega)$ as
\begin{equation}
    \hat g(\omega)=\hat g_*(\omega)+\hat a(\omega)\qquad\text{where }\hat g_*(\omega)\coloneqq2\hat{\bar\Omega}_i(\omega)\Theta(\omega),
\end{equation}
and $\Theta$ is the Heaviside step function and $\hat a$ is an arbitrary anti-Hermitian function. That is, we are free to choose $\hat a$ such that we minimise the spectral widths of $\bar I_i$ and $\bar Q_i$. We define the spectral width of a function as
\begin{equation}
    \operatorname{SW}(\hat h)\coloneqq\sup\{\left|\omega\right|:\hat h(\omega)\ne 0\}.
\end{equation}
As $\bar I_i$ is Hermitian and $i\bar Q_i$ is anti-Hermitian, then 
\begin{equation}
    w\coloneqq\max\left\{\operatorname{SW}(\hat{\bar I}_i), \operatorname{SW}(\hat{\bar Q}_i)\right\}=\operatorname{SW}[\hat g(\omega+1)]=\max\left\{\operatorname{SW}[\hat g_*(\omega+1)],\operatorname{SW}[\hat a(\omega+1)]\right\}
\end{equation}
As $\hat a$ is anti-Hermitian $\operatorname{SW}[\hat a(\omega+1)]=\operatorname{SW}[\hat a(\omega)]+1$. Further, $\operatorname{SW}[\hat g_*(\omega+1)]=\max\left\{1,\operatorname{SW}[\hat{\bar\Omega}_i]-1\right\}$. Therefore
\begin{equation}
    w=\max\left\{\operatorname{SW}[\hat a]+1,\operatorname{SW}[\hat{\bar\Omega}_i]-1\right\}
\end{equation}
Thus, any $\hat a$ such that $\operatorname{SW}[\hat a]\le\max\left\{0,\operatorname{SW}[\hat{\bar\Omega}_i]-2\right\}$ minimises $w$. Thus, $\hat a(\omega)=0$ is always an optimal value, and the only optimal solution for $\operatorname{SW}[\hat{\bar\Omega}_i]\le2$.

By taking the Hermitian and anti-Hermitian components of this solution and using $\hat{\bar\Omega}_i(\omega)=\hat{\bar\Omega}^*_i(-\omega)$, we can extract $\hat{\bar I}_i(\omega)$ and $\hat{\bar Q}_i(\omega)$ as
\begin{equation}
    \begin{bmatrix}
        \hat{\bar I}_i(\omega)\\
        \hat{\bar Q}_i(\omega)
    \end{bmatrix}=\begin{bmatrix}
    1&\phantom{-}1\\
    i&-i
    \end{bmatrix}\begin{bmatrix}
    \hat{\bar\Omega}_i(\omega+1)\Theta(1+\omega)\\
    \hat{\bar\Omega}_i(\omega-1)\Theta(1-\omega)
    \end{bmatrix}+\frac{1}{2}\begin{bmatrix}
        1&-1\\
        i&i
    \end{bmatrix}\begin{bmatrix}
        \hat a(\omega+1)\\
        \hat a(\omega-1)
    \end{bmatrix}.
\end{equation}
Consider $\omega\ge\max\left\{1, \operatorname{SW}(\hat a)-1\right\}$:
\begin{equation}
    \begin{bmatrix}
        \hat{\bar I}_i(\omega)\\
        \hat{\bar Q}_i(\omega)
    \end{bmatrix}=\begin{bmatrix}
    \hat{\bar\Omega}_i(\omega+1)-\frac{1}{2}\hat a(\omega-1)\\
    i\hat{\bar\Omega}_i(\omega-1)+\frac{1}{2}\hat a(\omega-1)
    \end{bmatrix}.
\end{equation}
Taking the $\ell^2$ norm we find:
\begin{equation}
    \left|\hat{\bar I}_i(\omega)\right|^2+\left|\hat{\bar Q}_i(\omega)\right|^2=2\left|\hat{\bar\Omega}_i(\omega+1)\right|^2+\frac{1}{2}\left|\hat a(\omega-1)\right|^2
\end{equation}
Thus, $\hat a(\omega-1)\ne0$ for $\omega\ge\max\left\{1, \operatorname{SW}(\hat a)-1\right\}$ will only increase the $\left|\hat{\bar I}_i(\omega)\right|^2+\left|\hat{\bar Q}_i(\omega)\right|^2$. Thus, we can minimise $\left|\hat{\bar I}_i(\omega)\right|^2+\left|\hat{\bar Q}_i(\omega)\right|^2$ for all $\omega\ge\max\left\{1, \operatorname{SW}(\hat a)-1\right\}$ by setting $\hat a(\omega-1)=0$ for all $\omega\ge\max\left\{1, \operatorname{SW}(\hat a)-1\right\}$. However, if $\hat a(\omega)=0$ for all $\omega\ge\max\left\{0, \operatorname{SW}(\hat a)-2\right\}$, then $\operatorname{SW}(\hat a)=0$ and $\hat a(\omega)=0$. Thus, we will proceed with $\hat a(\omega)=0$ as $\hat a(\omega)=0$ minimises $w$ and because while any other $\hat a$ may decrease $\left|\hat{\bar I}_i(\omega)\right|^2+\left|\hat{\bar Q}_i(\omega)\right|^2$ for some specific value of $\omega$ it will do this at the price of increasing $\left|\hat{\bar I}_i(\left|\omega\right|+2)\right|^2+\left|\hat{\bar Q}_i(\left|\omega\right|+2)\right|^2$. Thus, our solution for $\hat{\bar I}_i$ and $\hat{\bar Q}_i$ in Fourier space is
\begin{equation}
    \begin{bmatrix}
        \hat{\bar I}_i(\omega)\\
        \hat{\bar Q}_i(\omega)
    \end{bmatrix}=\begin{bmatrix}
    1&\phantom{-}1\\
    i&-i
    \end{bmatrix}\begin{bmatrix}
    \hat{\bar\Omega}_i(\omega+1)\Theta(1+\omega)\\
    \hat{\bar\Omega}_i(\omega-1)\Theta(1-\omega)
    \end{bmatrix}.
\end{equation}
Fourier transforming back, we find
\begin{align}
    \begin{bmatrix}
        I_i\left(t\right)\\
        Q_i\left(t\right)
    \end{bmatrix}&=\frac{1}{\mu_i}\begin{bmatrix}
        \cos\phi_i(t)&-\sin\phi_i(t)\\
        \sin\phi_i(t)&\phantom{-}\cos\phi_i(t)
    \end{bmatrix}\begin{bmatrix}
        \Omega_i\left(t\right)\\
        \displaystyle\frac{1}{\pi}\operatorname{p.v.}\int\limits_{-\infty}^\infty\dd{x}\frac{\Omega_i\circ\phi_i^{-1}(x)}{x-\phi_i(t)}
    \end{bmatrix}\\
    &=\frac{1}{\mu_i}\begin{bmatrix}
        \cos\phi_i(t)&-\sin\phi_i(t)\\
        \sin\phi_i(t)&\phantom{-}\cos\phi_i(t)
    \end{bmatrix}\begin{bmatrix}
        \Omega_i\left(t\right)\\
        \displaystyle\frac{1}{\pi}\operatorname{p.v.}\int\limits_{-\infty}^\infty\dd{x}\dot\phi_i(x)\frac{\Omega_i(x)}{\phi_i(x)-\phi_i(t)}
    \end{bmatrix},
\end{align}
where $\operatorname{p.v.}$ is the Cauchy principal value.

For the global case, we can follow a similar procedure. To make progress, we will need to restrict $\phi_j(t)=\omega_j\phi(t)$ for all $j$. Taking Let $\bar\Omega\coloneqq\frac{1}{\mu_i}\Omega_i\circ\phi^{-1}$, $\bar I_i\coloneqq I_i\circ\phi^{-1}$, and $\bar Q_i\coloneqq Q_i\circ\phi^{-1}$ we find
\begin{equation}
    \bar\Omega(t)=\sum_j\left[\bar I_j(t)\cos(\omega_jt)+\bar Q_j(t)\sin(\omega_jt)\right],
\end{equation}
and taking the Fourier transform (assuming $\phi_i(t)\to\pm\infty$ as $t\to\pm\infty$ or $\Omega_i(t)$ vanishes outside a set of finite measure), we find
\begin{equation}
    2\hat{\bar\Omega}(\omega)=\sum_j\left[\hat{\bar I}_j(\omega-\omega_j)-i\hat{\bar Q}_j(\omega-\omega_j)+\hat{\bar I}_j(\omega+\omega_j)+i\hat{\bar Q}_j(\omega+\omega_j)\right].
\end{equation}
We can identify the functions $\hat g_j(\omega)=\hat{\bar I}_j(\omega-\omega_j)-i\hat{\bar Q}_j(\omega-\omega_j)$ such that the only constraint on the $\left\{\hat g_j(\omega)\right\}$ is that the Hermitian component of $\sum_j\hat g_j(\omega)$ is $\hat{\bar\Omega}(\omega)$. As for the local case,
\begin{equation}
    w_j\coloneqq\max\left\{\operatorname{SW}(\hat{\bar I}_j), \operatorname{SW}(\hat{\bar Q}_j)\right\}=\operatorname{SW}[\hat g(\omega+\omega_j)].
\end{equation}
Thus, we can minimise $w_j$ by taking
\begin{equation}
    \hat g_j(\omega)=2\hat{\bar\Omega}(\omega)\sqcap_j(\omega)\qquad\text{where }\sqcap_j\left(\omega\right)=\begin{cases}
        \Theta(\omega)-\Theta(\omega-\frac{1}{2}[\omega_1+\omega_2])&j=1,\\
        \Theta(\omega-\frac{1}{2}[\omega_{j-1}+\omega_j])-\Theta(\omega-\frac{1}{2}[\omega_j+\omega_{j+1}])&2\le j\le n-1,\\
        \Theta(\omega-\frac{1}{2}[\omega_{n-1}+\omega_n])&j=n.
    \end{cases}
\end{equation}
As for the local case, by taking the Hermitian and anti-Hermitian components of this solution and using $\hat{\bar\Omega}(\omega)=\hat{\bar\Omega}^*(-\omega)$, we can extract $\hat{\bar I}_j(\omega)$ and $\hat{\bar Q}_j(\omega)$ as
\begin{equation}
    \begin{bmatrix}
        \hat{\bar I}_j(\omega)\\
        \hat{\bar Q}_j(\omega)
    \end{bmatrix}=\begin{bmatrix}
    1&\phantom{-}1\\
    i&-i
    \end{bmatrix}\begin{bmatrix}
    \hat{\bar\Omega}(\omega+\omega_j)\sqcap_j(\omega_j+\omega)\\
    \hat{\bar\Omega}(\omega-\omega_j)\sqcap_j(\omega_j-\omega)
    \end{bmatrix},
\end{equation}
and Fourier transforming back, we find
\begin{multline}
    \begin{bmatrix}
        I_j(t)\\
        Q_j(t)
    \end{bmatrix}=\frac{1}{2\pi\mu_i}\begin{bmatrix}
        \cos(\omega_jt)&-\sin(\omega_jt)\\
        \sin(\omega_j t)&\phantom{-}\cos(\omega_jt)
    \end{bmatrix}\int\limits^\infty_{-\infty}\dd{t'}\Omega_i(t')K(t-t')\\
    \text{where }K(t)\coloneqq\begin{cases}
        (\omega_1+\omega_2)\begin{bmatrix}
            \operatorname{sinc}(\frac{1}{2}[\omega_1+\omega_2]t)\\
            \operatorname{sinc}(\frac{1}{4}[\omega_1+\omega_2]t)\sin(\frac{1}{4}[\omega_1+\omega_2]t)
        \end{bmatrix}&j=1,\\
        (\omega_{j+1}-\omega_{j-1})\operatorname{sinc}(\frac{1}{4}[\omega_{j+1}-\omega_{j-1}]t)\begin{bmatrix}
            \cos(\frac{1}{4}[\omega_{j+1}+\omega_{j+1}+2\omega_j]t)\\
            \sin(\frac{1}{4}[\omega_{j+1}+\omega_{j+1}+2\omega_j]t)
        \end{bmatrix}&2\le j\le n-1,\\
        \begin{bmatrix}
        \phantom{-}\cos(\frac{1}{2}[\omega_{n-1}+\omega_n]t)&\sin(\frac{1}{2}[\omega_{n-1}+\omega_n]t)\\
        -\sin(\frac{1}{2}[\omega_{n-1}+\omega_n]t)&\cos(\frac{1}{2}[\omega_{n-1}+\omega_n]t)
        \end{bmatrix}\begin{bmatrix}
            \pi\delta(t)\\
            \operatorname{p.v.}\frac{1}{t}
        \end{bmatrix}&j=n.
    \end{cases}
\end{multline}

\section{Semiconducting spin qubit virtual \texorpdfstring{$Z$}{Z} pulse infidelities}
\label{app: infidelities}

Due to the rotating wave approximation, the infidelities of the virtual $Z$ pulses are limited when we include single-qubit drives. We numerically simulate the system's evolution using PySTE \cite{Long2025, Long_PySTE_2025} to estimate an order of magnitude of the infidelities. To obtain our device parameters, we rescale the pulses in \cref{fig: example} such that the qubit frequencies are $18$ and $18.03$ GHz, giving a pulse duration of $33.\dot3$ ns. We also rescale $J$ by $10$ MHz and $I$ and $Q$ by $2.6$ MHz. The infidelities of the implemented pulse to the desired evolutions are presented in \cref{fig: infidelity}. The infidelity follows a power law scaling in the virtual $Z$ pulse amplitude. The exponent is close to $2$.

\begin{figure*}
    \includegraphics{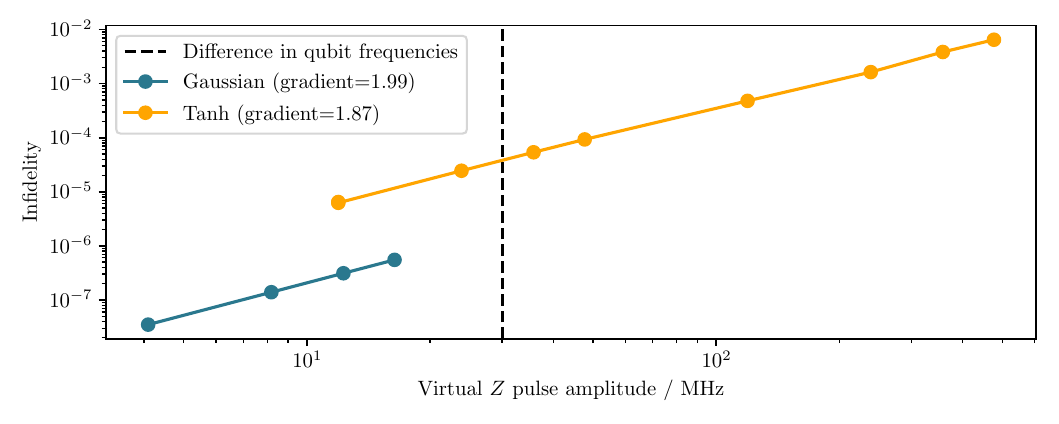}
    \caption{\textbf{Infidelity of pulse sequence implementing virtual $Z$ pulse to the desired evolution introduced by the rotating wave approximation.} The infidelity is plotted against the maximum magnitude of the virtual $Z$ pulse on a log-log scale for two different pulse shapes. We label the pulse shapes by the function they integrate to. Gaussian corresponds to the left-hand column of \cref{fig: example} and Tanh corresponds to the remaining two columns of \cref{fig: example}. As a reference value, the difference in qubit frequencies is displayed as a vertical dashed line. The gradients of the two lines on the log-log scale are 1.99 and 1.87, respectively.}
    \label{fig: infidelity}
\end{figure*}

\section{Solving for the rotating frame}
\label{app: flux solution}

In the body of the article, we encountered the following system of equations:
\begin{subnumcases}{}
    \text{\cref{eq: f equations}}\tag{\ref{eq: ft constraints}a},\\
    \begin{aligned}
        &\tilde\omega_k\circ\Phi_k'(\tau)-\dot\phi_k(\tau)\\&=\!\dv{f}{\tau}\!\!\left[\tilde\omega_k\!\circ\!\Phi_k\!\circ\! f(\tau)\!-\!\dot\phi_k\circ f(\tau)\right]\forall k\in\left\{i, j\right\},
    \end{aligned}\tag{\ref{eq: ft constraints}b}\\
    J_{ij}\left(\Phi_i'\!\left(\tau\right),\Phi_j'\!\left(\tau\right)\right)\!=\!\dv{f}{\tau}J_{ij}\left(\Phi_i\!\circ\! f\!\left(\tau\right),\Phi_j\!\circ\! f\!\left(\tau\right)\right),\tag{\ref{eq: ft constraints}c}
\end{subnumcases}
where \cref{eq: f equations} is:
\begin{subnumcases}{}
    \textrm{all $f\left(\tau\right)$ are solutions}&$\mathcal E_{ij}\left(0\right)\in\operatorname{span}_{\mathbb R}\mathcal Z$,\tag{\ref{eq: f equations}a}\\
    \Delta^{\pm}\phi_{ij}\left(\tau\right)=\Delta^{\pm}\phi_{ij}\circ f\left(\tau\right)+\tfrac{1}{\hbar}\Delta^{\pm}V_{ij}\circ f\left(\tau\right)\mod 4\pi&$\mathcal E_{ij}\left(0\right)\in\operatorname{span}_{\mathbb R}\mathcal Z\cup\mathcal C_{\pm}$,\tag{\ref{eq: f equations}b}\\
    \phi_k\left(\tau\right)=\phi_k\circ f\left(\tau\right)+\tfrac{1}{\hbar}V_k\circ f\left(\tau\right)\mod 4\pi&$\mathcal E_{ij}\left(0\right)\in\operatorname{span}_{\mathbb R}\mathcal Z\cup\mathcal Q_k$,\tag{\ref{eq: f equations}c}\\
    \begin{cases}
        \phi_i\left(\tau\right)=\phi_i\circ f\left(\tau\right)+\tfrac{1}{\hbar}V_i\circ f\left(\tau\right)\mod 4\pi\\
        \phi_j\left(\tau\right)=\phi_j\circ f\left(\tau\right)+\tfrac{1}{\hbar}V_j\circ f\left(\tau\right)\mod 4\pi
    \end{cases}&otherwise,\tag{\ref{eq: f equations}d}
\end{subnumcases}

For the case $\mathcal E_{ij}\left(0\right)\in\operatorname{span}_{\mathbb R}\mathcal Z$ then the exists the trivial solution:
\begin{subnumcases}{}
    f(\tau)=\tau,\\
    \phi_i(\tau)\text{ is any function},\\
    \phi_j(\tau)\text{ is any function},\\
    \Phi_i'(\tau)=\Phi_i(\tau),\\
    \Phi_j'(\tau)=\Phi_j(\tau).
\end{subnumcases}

If $\mathcal E_{ij}\left(0\right)\in\operatorname{span}_{\mathbb R}\mathcal Z\cup\mathcal C_{\pm}$ or $\mathcal E_{ij}\left(0\right)\in\operatorname{span}_{\mathbb R}\mathcal Z\cup\mathcal Q_k$ we will proceed by differentiate \cref{eq: f equations} to find
\begin{subnumcases}{}
    \Delta^{\pm}\dot\phi_{ij}\left(\tau\right)=\dv{f}{\tau}\left[\Delta^{\pm}\dot\phi_{ij}\circ f\left(\tau\right)+\tfrac{1}{\hbar}\Delta^{\pm}v_{ij}\circ f\left(\tau\right)\right]&$\mathcal E_{ij}\left(0\right)\in\operatorname{span}_{\mathbb R}\mathcal Z\cup\mathcal C_{\pm}$,\\
    \dot\phi_k\left(\tau\right)=\dv{f}{\tau}\left[\dot\phi_k\circ f\left(\tau\right)+\tfrac{1}{\hbar}v_k\circ f\left(\tau\right)\right]&$\mathcal E_{ij}\left(0\right)\in\operatorname{span}_{\mathbb R}\mathcal Z\cup\mathcal Q_k$.
\end{subnumcases}
By subtracting these from an appropriate linear combination of Eq.~(\ref{eq: ft constraints}b) for $k\in\left\{i,j\right\}$ we find to two new equations:
\begin{subnumcases}{\label{eq: pm flux}}
    \displaystyle\Delta^{\pm}\tilde\omega_k\circ\Phi_k'(\tau)=\dv{f}{\tau}\left[\Delta^{\pm}\tilde\omega_k\circ\Phi_k\circ f(\tau)+\tfrac{1}{\hbar}\Delta^{\pm}v_k\circ f\left(\tau\right)\right],\\
    \displaystyle\Delta^{\mp}\tilde\omega_{\bar k}\circ\Phi_{\bar k}'(\tau)-\Delta^{\mp}\dot\phi_{\bar k}(\tau)=\dv{f}{\tau}\left[\Delta^{\mp}\tilde\omega_{\bar k}\circ\Phi_{\bar k}\circ f(\tau)-\Delta^{\mp}\dot\phi_{\bar k}\circ f(\tau)\right],
\end{subnumcases}
if $\mathcal E_{ij}\left(0\right)\in\operatorname{span}_{\mathbb R}\mathcal Z\cup\mathcal C_{\pm}$, and
\begin{subnumcases}{\label{eq: no change flux}}
    \displaystyle\tilde\omega_k\circ\Phi_k'(\tau)=\dv{f}{\tau}\left[\tilde\omega_k\circ\Phi_k\circ f(\tau)+\tfrac{1}{\hbar}v_k\circ f\left(\tau\right)\right]\\
    \displaystyle\tilde\omega_{\bar k}\circ\Phi_{\bar k}'(\tau)-\dot\phi_{\bar k}(\tau)=\dv{f}{\tau}\left[\tilde\omega_{\bar k}\circ\Phi_{\bar k}\circ f(\tau)-\dot\phi_{\bar k}\circ f(\tau)\right]
\end{subnumcases}
if $\mathcal E_{ij}\left(0\right)\in\operatorname{span}_{\mathbb R}\mathcal Z\cup\mathcal Q_k$.

Notice we have now decoupled our equations such that, given a function $f(\tau)$, we can compute every other function using a series of equations that only contain one unknown function. We represent this in the diagrams below, where each function at a vertex can be computed using the functions at the start of all incoming arrows using the equation that labels them:
\begin{equation}
    \begin{tikzcd}
    f\arrow[d, "\text{\cref{eq: f equations}}", leftrightarrow]\arrow[r, "\text{Eq.~(\ref{eq: pm flux}a)}", leftrightarrow]\arrow[rr, "\text{Eq.~(\ref{eq: ft constraints}c)}", bend left=50]\arrow[rrr, "\text{Eq.~(\ref{eq: ft constraints}c)}", bend left=60]&\Delta^{\pm}\tilde\omega\circ\Phi'\arrow[r, "\text{Eq.~(\ref{eq: ft constraints}c)}"]&\Delta^{\mp}\tilde\omega\circ\Phi'\arrow[r, "\text{Eq.~(\ref{eq: pm flux}b)}"]&\Delta^{\mp}\phi_{ij}\\
    \Delta^{\pm}\phi_{ij}
    \end{tikzcd}
\end{equation}
if $\mathcal E_{ij}\left(0\right)\in\operatorname{span}_{\mathbb R}\mathcal Z\cup\mathcal C_{\pm}$, and
\begin{equation}
    \begin{tikzcd}
    f\arrow[d, "\text{\cref{eq: f equations}}", leftrightarrow]\arrow[r, "\text{Eq.~(\ref{eq: no change flux}a)}", leftrightarrow]\arrow[rr, "\text{Eq.~(\ref{eq: ft constraints}c)}", bend left=70]\arrow[rrr, "\text{Eq.~(\ref{eq: ft constraints}c)}", bend left=80]&\tilde\omega_k\circ\Phi_k'\arrow[r, "\text{Eq.~(\ref{eq: ft constraints}c)}"]&\tilde\omega_{\bar k}\circ\Phi_{\bar k}'\arrow[r, "\text{Eq.~(\ref{eq: no change flux}b)}"]&\phi_{\bar k}\\
    \phi_k
    \end{tikzcd}
\end{equation}
if $\mathcal E_{ij}\left(0\right)\in\operatorname{span}_{\mathbb R}\mathcal Z\cup\mathcal Q_k$.

Notice we can compute the frequencies (or fluxes) from $f$ without ever needing to compute the rotating frame. As we have four equations for five functions, we are free to choose one function. For ease, we choose $f$. It is straightforward to then solve for $\Delta^{\pm}\tilde\omega\circ\Phi'$ or $\tilde\omega_k\circ\Phi_k'$, respectively, by evaluating the functions at each time of interest $\tau$. However, it is harder to solve \cref{eq: f equations} as both depend on the unknown function at multiple points in time. Fortunately, if we choose $f(\tau)\le\tau$ or $f(\tau)\ge\tau$ then we can solve \cref{eq: f equations} as follows: We will consider the case $\mathcal E_{ij}\left(0\right)\in\operatorname{span}_{\mathbb R}\mathcal Z\cup\mathcal Q_k$ for simplicity but the method can be used for \cref{eq: f equations} with $\mathcal E_{ij}\left(0\right)\in\operatorname{span}_{\mathbb R}\mathcal Z\cup\mathcal C_{\pm}$. First, we generate an array of times $\tau_{i+1}=f(\tau_i)$ recursively. If $f(\tau)\le\tau$ then $\tau_{i+1}$ increases monotonically. Alternatively, if $f(\tau)\ge\tau$ then $\tau_{i+1}$ decreases monotonically. Substituting $\tau=\tau_{i}$ and $f(\tau_i)=\tau_{i+1}$ into \cref{eq: f equations} we find
\begin{equation}
    \phi_k\left(\tau_{i}\right)=\phi_k\left(\tau_{i+1}\right)+\tfrac{1}{\hbar}V_k\left(\tau_{i+1}\right)\mod 4\pi.
\end{equation}
Now we can evaluate $\phi_k$ for all $\tau_i$ recursively. Thus, given an initial point $\tau_0$ and either boundary condition $\phi_k(\tau_0)$ or $\phi_k(\tau_{\max})$ we can evaluate $\phi_k$ at a finite number of points. We can then linearly interpolate these points for a numerical solution. To increase the resolution of our numerical solution, we need only repeat this process with a different $\tau_0$ to obtain a new array of points and so more values of $\phi_k$ between which we can interpolate. To choose the boundary conditions, we can employ one of two procedures. If $f(\tau)\ge\tau$ then we can choose $\tau_0$ sufficiently small that $\phi_k(\tau_0)\approx0$. Alternatively, if $f(\tau)\le\tau$ we note that the sequence $\tau_{i+1}=f(\tau_i)$ will decrease to a value $\tau_*$ for which $f(\tau_*)=\tau_*$. We can then take sufficiently many recursive steps so that we can use $\phi_k(\tau_{\max})\approx\phi_k(\tau_*)=\tau_*$.

Next, we turn to solving Eq.~(\ref{eq: ft constraints}c). This will depend on the exact form of the coupling $J_{ij}(\tilde\omega_i,\tilde\omega_j)$. Here we consider two examples: (i) For directly-coupled capacitive qubits $J_{ij}(\tilde\omega_i,\tilde\omega_j)\propto\sqrt{\tilde\omega_i\tilde\omega_j}$ \cite{10.1063/1.5089550}. (ii) For capacitive couplings mediated through bus resonators $J_{ij}(\tilde\omega_i,\tilde\omega_j)\propto\frac{1}{\tilde\omega_i-\omega_r}+\frac{1}{\tilde\omega_j-\omega_r}$ where $\omega_r$ is the bus resonator frequency \cite{10.1063/1.5089550}. Both of these leave Eq.~(\ref{eq: ft constraints}c) as a quadratic equation for the unknown function in terms of the known functions, which can readily be solved.

Finally, we can solve Eqs.~(\ref{eq: no change flux}b) and (\ref{eq: pm flux}b) with the method developed for \cref{eq: f equations} before numerically integrating the resulting function.

\section{Dilation optimisation}
\label{app: dilation optimisation}

To optimise our dilation to minimise pulse distortion, in \cref{fig: ft sc example}, we utilise the following cost functional:
\begin{equation}
    C[f;\lambda]\coloneqq\frac{1}{\left|f(T)-f(0)\right|}\int\limits_{f^{-1}(0)}^{f^{-1}(T)}\dd{\tau}\left[\sqrt{\sum_{i=1}^2\left[\tilde\omega'_i(\tau)-\tilde\omega_i(\tau)\right]^2}+\lambda_0\operatorname{ReLU}(\tau-f(\tau))+\sum_{i=1}^2\lambda_i\operatorname{ReLU}(\omega_{i0}-\tilde\omega'_i(\tau))\right],
\end{equation}
where $\operatorname{ReLU}(x)\coloneqq \max\left\{0,x\right\}$. The first term minimises the pulse distortion. The second term employs a Lagrange multiplier $\lambda_0$ to enforce $f(\tau)\ge\tau$. The third term employs the Lagrange multipliers $\lambda_1$ and $\lambda_2$ to ensure the distorted qubit frequencies $\left\{\tilde\omega'_i(\tau))\right\}_{i=1}^2$ are greater than the minimum qubit frequencies $\left\{\omega_{i0}\right\}_{i=1}^2$. To discretise our solution space, we take $f(\tau)$ to have the following form,
\begin{equation}
    f(\tau)=\tau+a\tau^2+b\tau^3+\sum_{m=4}^Mc_m\tau^m,
\end{equation}
where $\tau_1$ and $\left\{c_m\in\mathbb R\right\}_{m=4}^M$ are free parameters and
\begin{align}
    a&=\frac{3-3A-\tau_1(B-3)}{\tau_1^2} & b&=\frac{2A-\tau_1(B-2)-2}{\tau_1^3}
\end{align}
and
\begin{align}
    A&=\sum_{m=4}^Mc_m\tau_1^m & B&=\sum_{m=4}^Mmc_m\tau_1^{m-1}.
\end{align}
These hard constrainsts force $\dv{f}{\tau}=1$ at $\tau=f^{-1}(0),f^{-1}(T)$ and $T=f(\tau_1)$. We choose $M=11$ for our simulations. This yields a 12-dimensional parameter space.

To perform the optimisation, we utilise COBYLA \cite{powellDirectSearchOptimization1994,zaikun_zhang_2023_8052655} with the additional constraints that $\lambda_i\ge 1$ for $i\in\left\llbracket0,2\right\rrbracket$.

\section{Normalizer group of the Pauli-\texorpdfstring{$Z$}{Z} rotations}
\label{app: normalizer}

In this appendix, we find the generators of the normaliser group of the Pauli-$Z$ rotations.

\begin{lemma}
Let $L\coloneqq\exp(i\operatorname{span}_{\mathbb R}\left\{Z_i\right\}_{i=1}^n)$. The normalizer group $\mathcal N\coloneqq\left\{u\in\operatorname{U}(2^n):uL=Lu\right\}$ is generated by $X_1$, $\left\{\operatorname{SWAP}_{i,i+1}\right\}_{i=1}^{n-1}$, and arbitrary phase gates, $P\coloneqq\exp(i\operatorname{span}_{\mathbb R}\left\{I,Z\right\}^{\otimes n})$.
\end{lemma}
\begin{proof}
In the computational basis, all elements of $L$ are diagonal. Thus for some $u\in\operatorname{U}(2^n)$ and $l\in L$ we find in the computational basis that
\begin{equation}
    [lul']_{ab}\equiv\sum_{c,d=1}^{2^n}l_{ac}u_{cd}l_{db}=l_{aa}u_{ab}l_{bb}'\qquad\forall a,b\in\left\llbracket2^n\right\rrbracket.
\end{equation}
Thus, as $L$ is a group, for all $u\in\mathcal N$ and $l\in L$ there exists $l'\in L$ such that
\begin{equation}
    l_{aa}u_{ab}l_{bb}'=u_{ab}\qquad\forall a,b\in\left[1,2^n\right].
\end{equation}
Thus, for all $a,b\in\left\llbracket2^n\right\rrbracket$ such that $u_{ab}\ne0$ we find $l_{bb}'=1/l_{aa}$. Therefore, $u\in\mathcal N$ can only have a single non-zero entry per column in the computational basis. As $\mathcal N\leq\operatorname{U}(2^n)$, then $\mathcal N$ must be a subgroup of the group $G$ generated by $P$ and the $2^n\times 2^n$ permutation matrices, $\operatorname{S_{2^n}}$. As $P$ is an abelian group of which $L$ is a subgroup, then $P\le\mathcal N$. Thus, $\mathcal N$ is generated by $P$ and a subgroup $H$ of $\operatorname{S_{2^n}}$.

To find $H\le\operatorname{S_{2^n}},G$ we first note that any element $u\in\mathcal N\le G$ can be expressed as $u=ps$ where $p\in P$ and $s\in\operatorname{S_{2^n}}$. Thus, substituting into $lul'=u$ and using the fact $p$ commutes with $l$ and $l'$ we find
\begin{align}
    \mathcal N\coloneqq P\cdot\left\{s\in\operatorname{S_{2^n}}:sL=Ls\right\},
\end{align}
\textit{i.e.}, $H=\left\{s\in\operatorname{S_{2^n}}:sL=Ls\right\}$ and $\mathcal N=P\cdot H$. Next by expressing the elements of $L$ as $\exp(i\sum_{i=1}^n\theta_iZ_i)$ where $\theta_i\in\left[0,2\pi\right)$ we can express the constraint on $s\in\operatorname{S_{2^n}}$ as
\begin{equation}
    \forall\vec\theta\in\left[0,2\pi\right)^n,\ \exists \phi\in\left[0,2\pi\right)^n\text{ such that }s\exp(i\sum_{i=1}^n\theta_iZ_i)s^{-1}=\exp(i\sum_{i=1}^n\phi_iZ_i).
\end{equation}
Manipulating the constraint, we find
\begin{align}
    s\exp(i\sum_{i=1}^n\theta_iZ_i)s^{-1}&=\exp(i\sum_{i=1}^n\phi_iZ_i)\\
    \iff\sum_{i=1}^n(-1)^{[s(x)]_i}\theta_i&=\sum_{i=1}^n(-1)^{x_i}\phi_i\mod 2\pi\qquad\forall x\in\mathbb Z_2^n\\
    &=-\sum_{i=1}^n(-1)^{\bar x_i}\phi_i\mod 2\pi\qquad\forall x\in\mathbb Z_2^n\\
    &=-\sum_{i=1}^n(-1)^{[s(\bar x)]_i}\theta_i\mod 2\pi\qquad\forall x\in\mathbb Z_2^n\\
    &=\sum_{i=1}^n(-1)^{\overline{[s(\bar x)]_i}}\theta_i\mod 2\pi\qquad\forall x\in\mathbb Z_2^n,
\end{align}
where $\bar x_i$ is the complement bit to $x_i$ and $s(x)$ is the action of the permution $s$ on the set of $n$-bitstring. As the constraint must hold for all $\vec\theta$ we conclude that $\overline{[s(x)]_i}=[s(\bar x)]_i$ for all $n$-bitstrings $x$. Therefore, only permutations $s$ that permute the bits and flip the bits satisfy the constraint. Thus, $H$ is a subgroup of the group generated by $X_1$ and $\left\{\operatorname{SWAP}_{i,i+1}\right\}_{i=1}^{n-1}$. As $X_1L=LX_1$ and $\operatorname{SWAP}_{i,i+1}L=L\operatorname{SWAP}_{i,i+1}$ for all $i\in\left[1,n-1\right]$ we conclude that $H$ is generated by $X_1$ and $\left\{\operatorname{SWAP}_{i,i+1}\right\}_{i=1}^{n-1}$.

\end{proof}

\end{document}